\documentclass[numbook,envcountsect,envcountsame,envcountreset,runningheads,smallextended]{svjour3}

\smartqed  
\usepackage{graphicx}
\usepackage{amsfonts,amssymb}
\usepackage{color}

\newcommand{\mm}{m}

\def \Rbrack {[\![}
\def \Lbrack {]\!]}

\def \Rbrack {[\![}
\def \Lbrack {]\!]}
\newcommand{\is}{\centerdot}

\title{Non-Arbitrage under a Class of Honest Times}

\author{  Anna Aksamit\and Tahir Choulli\and Jun Deng\and Monique Jeanblanc
}
\institute{
 Anna  Aksamit\at
 Mathematical Institute, University of Oxford,\\
Oxford, United Kingdom\\
Monique Jeanblanc\at
Laboratoire de Math\'ematiques et Modelisation d'\'Evry (LaMME), Universit\'e d'Evry  Val
d'Essonne, UMR CNRS 8071, France\\
Tahir Choulli (corresponding author)\at
Mathematical and Statistical Sciences Depart.,
University of Alberta, Edmonton, Canada \\
              \email{tchoulli@ualberta.ca}           \\
Jun Deng \at
 School of Banking and Finance,\\
  University of International Business and Economics, Beijing, China}

\date{Version of April 1st, 2016}

\begin{document}
\titlerunning{Non-Arbitrage and Honest Times}
\maketitle

\begin{abstract}
This paper quantifies the interplay between the non-arbitrage
notion of No-Unbounded-Profit-with-Bounded-Risk (NUPBR hereafter)
and additional information generated by a random time. This study
complements the one of Aksamit/Choulli/Deng/Jeanblanc
\cite{aksamit/choulli/deng/jeanblanc} in which the authors studied
similar topics for the case of stopping at the random time
instead, while herein we are
concerned with the part after the occurrence of the random time.
Given that all the literature ---up to our knowledge--- proves
that the NUPBR notion is always violated after honest times
that avoid stopping times in a continuous filtration,  herein we
propose a {\it new class of honest times} for which the NUPBR
notion can be preserved for some models. For this family of
honest times, we elaborate two principal results. The first main
result characterizes the pairs of initial market and honest time
for which the resulting model preserves the NUPBR property, while
the second main result characterizes the honest times that
preserve the NUPBR property for any quasi-left continuous model. Furthermore, we construct
explicitly ``the-after-$\tau$" local martingale deflators for a
large class of initial models (i.e. models in the small
filtration) that are already risk-neutralized.
\end{abstract}


\newpage
\section{Introduction}
This paper complements the study undertaken in
\cite{aksamit/choulli/deng/jeanblanc} about quantifying the exact interplay between an extra information/uncertainty
and arbitrage for quasi-left-continuous models\footnote{A quasi-left-continuous model/process
is a process that does not jump on predictable stopping times}. Similarly as in \cite{aksamit/choulli/deng/jeanblanc}, we focus on the
non-arbitrage concept of No-Unbounded-Profit-with-Bounded-Risk
(NUPBR hereafter), and the extra information is the time of the occurrence of a random time, when it occurs. It is clearly stated
 in \cite{Choulli/Deng/Ma} (see \cite{aksamit/choulli/deng/jeanblanc}) that when the NUPBR is violated, none of the existing method for pricing
 and optimisation problems works. {\bf Throughout the paper, arbitrages means Unbounded-Profit-with-Bounded-Risk strategies}.
\subsection{What are the Main Goals and the Related Literature?} Throughout the paper, we consider given a stochastic
 basis $(\Omega, {\cal G},  {\mathbb  F },  P)$,  where $ {\cal F}_\infty:=\vee_{t\geq 0}{\cal F}_t
 \subseteq {\cal {G}}$, and the filtration ${\mathbb  F }:=({\cal F}_t)_{t\geq 0}$ satisfies the usual hypotheses
  (i.e. right continuity and completeness) and models the
flow of ``public" information that all agents receive through time. The initial
financial market is defined on this basis and is represented by a {\bf $d$-dimensional
semimartingale} $S$ and a riskless asset, with null interest rate.
In addition to this initial model,
 we consider a fixed random time (a non-negative random variable)
  denoted by $\tau$. This random time can represent the death time of an insurer, the default time of a firm,
   or any occurrence time of an influential event that can impact the market somehow. In this setting,
    our aim lies in answering the following.
 $$\mbox{If}\ (\Omega,{\mathbb  F },
S)\ \mbox{ is arbitrage-free, then what can be said about}\ (\Omega, {\mathbb  F },  S,\tau)?$$
 After modeling the {\it new informational system}, this question translates
into whether $(\Omega, {\mathbb  G },  S)$ has arbitrages or not. Here $\mathbb G$, that will be specified
mathematically in the next section, is the new flow of information
that incorporates the flow $\mathbb F$ and $\tau$, as soon as it
occurs, and  makes $\tau$ a $\mathbb G$-stopping time.
Thanks to \cite{Takaoka} (see also \cite{ChoulliStricker96} for
the continuous case and \cite{kardaras12} for the one dimensional
case), one can easily prove that $(\Omega, {\mathbb  G },  S)$
satisfies the NUPBR  condition if and only if both models
$(\Omega, {\mathbb G },  S^{\tau})$ and $(\Omega, {\mathbb  G },
S-S^{\tau})$ fulfill the NUPBR condition. In \cite{aksamit/choulli/deng/jeanblanc}, the authors focused on
$(\Omega, {\mathbb  G}, S^{\tau})$, while the second part
$(\Omega, {\mathbb  G},  S-S^{\tau})$ constitutes the main
objective of this paper. As it will be mathematically specified later,
the NUPBR notion consists, roughly speaking, of ``controlling" in
some sense the gain processes that are bounded uniformly in time and randomness from below by one.
 Mathematically speaking, these processes are stochastic
integrals with respect to the asset's price process. Thus, due to
the Dellacherie-Mokobodski criterion, the first obstacle in
investigating the NUPBR for $(\Omega, {\mathbb  G},
S-S^{\tau})$, lies in whether this model is an integrator for ``admissible" but complex (not only
buy-and-hold) financial strategies or not. This boils down to the model fulfilling the semimartingale property (see Theorem 80 in
\cite{dm2} page 401). Thus, the ``honest" assumption on $\tau$ guarantees the preservation of the semimartingale property after $\tau$.
It is known that (see \cite[Th\'eor\`eme 4.14]{Jeu}), in contrast to $({\mathbb  G },  S^{\tau})$,
 the semimartingale structures might fail for $({\mathbb  G },  S-S^{\tau})$
  when $\tau$ is arbitrary general. Therefore, for the rest of the paper,
   $\tau$ is assumed to be honest, a fact that will be mathematically defined in the next section. \\

\noindent Recently, in \cite{Imk} and \cite{fjs}, it is proved when honest times avoid stopping times and the filtration is Brownian that
 the NUPBR  property fails for $(S-S^{\tau},\mathbb G)$. Thus, our first goal is to answer the following
\begin{eqnarray}\label{quesion1}\mbox{Is there any $\tau$ for which  NUPBR  is preserved for some markets?}\end{eqnarray}
We answer this question positively, and we focus afterwards on quantifying the interplay between
$\tau$ and the initial market model that is responsible for arbitrages after $\tau$.
 This can be achieved, in our view, by finding a functional ${\cal K}$, that can be observed using the public information only, such that
 \begin{equation}\label{question2}
 \left({\cal K}(S),\mathbb F\right)\ \mbox{is arbitrage-free if and only if}\ (S-S^{\tau},\mathbb G)\  \mbox{does}.
 \end{equation}
\subsection{Our Financial and Mathematical Achievements}
Our first original contribution proposes a {\bf new class of honest times}
for which there are markets that preserve the NUPBR condition after $\tau$, and hence our first aim described in (\ref{quesion1}) is reached. Our
family of honest times includes all the $\mathbb F$-stopping times
as well as many examples of non $\mathbb F$-stopping times. By
considering this subclass of honest times throughout the paper,
our principal novelty resides in achieving our second aim of (\ref{question2}), and describe as explicit as possible
the functional ${\cal K}$. As a result, honest times belonging to our class might
 induce ``the-after-$\tau$'' arbitrages only if the initial market jumps.\\
 
\noindent This paper is organized as follows. In the following
section (Section \ref{SectionAfterRandomTimeQLC}), we present our
main results, their  immediate  consequences, and/or their
economic and financial interpretations. In this section, we also develop many examples and show how
 the main ideas came into play.  Section \ref{explicitdeflators}
deals with the derivation of
explicit local martingale deflators for a class of processes. The last section (Section
\ref{Sectionproofs}) focuses on proving the main
theorem and other related results announced. The paper contains also an appendix where some of
the existing and/or new technical results are summarized.


\section{The Main Results and their Financial Interpretations}\label{SectionAfterRandomTimeQLC}
This section contains three subsections. The first subsection defines notations and the NUPBR concept,
while the second subsection develops simple examples of informational markets and explains how some ingredients
 of the main results play natural and important r\^ oles. The last subsection announces the principal results,
  their applications, and gives their financial meanings as well.

\subsection{Notations and Preliminaries}
\noindent In what follows, $\mathbb H$ denotes a filtration satisfying the usual hypotheses.
 The set of $\mathbb H$-martingales  is denoted by ${\cal M}(\mathbb H)$. As usual, ${\cal A}^+(\mathbb H)$ denotes the set of increasing,
right-continuous, $\mathbb H$-adapted and integrable processes.\\
\noindent If ${\cal C}(\mathbb H)$ is a class of $\mathbb H$-adapted processes,
 we denote by ${\cal C}_0(\mathbb H)$ the set of processes $X\in {\cal C}(\mathbb H)$ with $X_0=0$, and
 by ${\cal C}_{loc}(\mathbb H)$
 the set  of processes $X$
 such that there exists a sequence $(T_n)_{n\geq 1}$ of $\mathbb H$-stopping times
  that increases to $+\infty$ and the stopped processes $X^{T_n}$
  belong
to ${\cal C}(\mathbb H)$. We put $ {\cal C}_{0,loc}={\cal
C}_0\cap{{\cal
C}}_{loc}$.\\
\noindent For a process $K$ with $\mathbb H$-locally integrable variation,
we denote by $K^{o,\mathbb H}$  its dual optional projection. The dual predictable projection of $K$   is denoted
$K^{p,\mathbb H}$. For a process $X$, we denote $^{o, \mathbb H
\!} X$ (resp.$\,^{p, \mathbb H\!} X$ ) its optional (resp.
predictable) projection with respect to $\mathbb H$.
\\
\noindent For a finite-dimensional $\mathbb H$-semimartingale
$X$, the set ${\cal{L}}(X, \mathbb H)$ is the set of $\mathbb
H$-predictable processes having the same dimension as $X$ and
being integrable w.r.t. $X$ and  for $H\in {\cal{L}}(X,
\mathbb H )$, the resulting integral is the one-dimensional
process denoted by $H\is X_t:= \int_0^t H_sdX_s$. {\bf Throughout
the paper, stochastic processes have arbitrary finite dimension}
(in case it is not specified). We recall the notion of
non-arbitrage that is addressed in this paper.

\begin{definition}\label{DefinitionofNUPBR}  An $\mathbb H$-semimartingale
$X$ satisfies the {\it No-Unbounded-Profit-with-Bounded-Risk} condition under $(\mathbb H,Q)$ if for any $T'\in (0,+\infty)$, the set
$$\label{boundedset}
{\cal K}_{T'} (X):=\displaystyle \Bigl\{(H\is X)_{T'} \ \big|\ \ \ H \in
{\cal{L}}(X,\mathbb H),\ \mbox{and}\ H\is X\geq -1\  \Bigr\}$$ is
bounded in probability under $Q$. Often, we abbreviate by saying that $X$ satisfies the NUPBR($\mathbb H, Q)$,
or the model $(X,\mathbb H,Q)$ satisfies the NUPBR. When $Q\sim P$, we simply drop the probability for short and simplifying the notations.
\end{definition}

\noindent For more details about this non-arbitrage
condition and its relationship to the literature, we refer the
reader to Aksamit et al. \cite{aksamit/choulli/deng/jeanblanc}.
The NUPBR property is intimately related to the existence of a
$\sigma$-martingale density. Below, we recall
the definition of $\sigma$-martingale and $\sigma$-martingale
density for a process.
 \begin{definition}\label{sigmsamartingaledensities} An $\mathbb H$-adapted process {$X$} is called
   an  $(\mathbb H,\sigma)$-martingale  if  there exists a real-valued $\mathbb H$-predictable process $\phi$ such that
 $$
 0<\phi\leq 1,\ \ \ \mbox{and $\phi\is X$ is an $\mathbb H$-martingale}.$$
If $X$ is $\mathbb H$-adapted, we call $(\mathbb H,\sigma)$-martingale density for $X$ ({\bf also called $\mathbb H$-deflator for $X$}),
 any real-valued positive
 $\mathbb H$-local martingale  $L$ such that $XL$ is an $(\mathbb H,\sigma)$-martingale.
  The set of all $\mathbb H$-deflators for $X$ is denoted by
\begin{equation}\label{Set4SigmaMgdensities}
{\cal L}_{\sigma}(X,{\mathbb H}):=\left\{ L\in {\cal
M}_{loc}(\mathbb H)\big|\ L>0,\ \ LX\ \mbox{ is an   $ (\mathbb
H,\sigma)$-martingale } \right\}.\end{equation}
\end{definition}
The equivalence between NUPBR($\mathbb H$) for a process $X$ and
${\cal L}_{\sigma}(X,{\mathbb H})\not=\emptyset$ is established
in \cite{aksamit/choulli/deng/jeanblanc}
(see Proposition 2.3) when the horizon may be infinite, and in \cite{Takaoka} for the case of finite horizon.\\

\noindent Beside the initial  model $\left(\Omega, \mathbb F, P, S\right)$ in which $S$
 is assumed to be a {\bf quasi-left-continuous semimartingale}, we  consider a random time $\tau$, to which we associate  the process
$D$ and the filtration $\mathbb G$ given by
$$\label{AandfiltrationG}
D:= I_{\Rbrack\tau,+\infty\Rbrack},\ \ \ \mathbb G=\left({\cal
G}_t\right)_{t\geq 0},\ \ \ {\cal G}_t :=
\displaystyle\bigcap_{s>t}\Bigl({\cal F}_s\vee \sigma(D_u, u\leq
s)\Bigr).$$ The filtration $\mathbb G$ is the smallest
right-continuous filtration which contains ${\mathbb  F }$ and
makes $\tau$ a stopping time. In the probabilistic literature,
$\mathbb G$ is called the progressive enlargement of $\mathbb F$
with $\tau$. In addition to $\mathbb G$ and $D$, we associate to
$\tau$ two important $\mathbb F$-supermartingales: the
$\mathbb F$-optional projection of $I_{\Rbrack 0, \tau\Rbrack}$
denoted $Z$, and the $\mathbb F$-optional projection of
$I_{\Rbrack 0, \tau\Lbrack}$, denoted $\widetilde Z$, which
satisfy
\begin{equation}\label{ZandZtilde}
Z_t := P\Bigl(\tau >t\ \big|\ {\cal F}_t\Bigr)\   \mbox{and}\ \ \ \widetilde Z_t:=P\left(\tau\geq t\ \Big|\ {\cal F}_t\right).
\end{equation}
$Z$ is right-continuous with left limits, while $\widetilde Z$ admits right limits and left
limits.  An important $\mathbb F$-martingale,  denoted by $m$, is
given by
 \begin{equation}\label{processmZ}
m := Z+D^{o,\mathbb F},\end{equation}
 where $D^{o,\mathbb F}$ is the $\mathbb F$-dual optional projection of
 $D=I_{\Rbrack\tau, \infty\Rbrack}$ (Note that   $Z$  is bounded
and $D^{o,\mathbb F}$ is nondecreasing and integrable).\\

\noindent To distinguish
the effect of filtration, we will denote $\langle .,
.\rangle^{\mathbb F},$ or $\langle ., . \rangle^{\mathbb G}$
 to specify the sharp bracket (predictable covariation process) calculated in the filtration ${\mathbb F}$ or
  ${\mathbb G},$ if confusion may rise. We recall that, for general semimartingales $X$ and $Y$, the
sharp bracket is (if it exists) the dual predictable projection of
the covariation process $[X,Y]$.   For the reader's
convenience, we recall the definition of honest time.

\begin{definition}\label{honesttime}  A random time $\sigma$ is honest, if for any $t$,
there exists an ${\cal {F}}_t$ measurable r.v. $\sigma _t$ such
that $\sigma I_{\{\sigma<t\}}=\sigma_t I_{\{\sigma<t\}}$.

\end{definition}
We refer to Jeulin \cite[Chapter 5]{Jeu} and Barlow \cite{Barlow}
for  more information about honest times. In this paper, we
restrict our study to the following subclass $\mathcal H$ of
random times:
\begin{equation}\label{assumptionOntau}
{\mathcal H}:= \{\tau\ \mbox{is an honest time satisfying}\ \ \
Z_{\tau}I_{\{\tau<+\infty\}}<1,\ \ \ \ P-a.s.\}\end{equation}
\begin{remark}
1) It is clear that any $\mathbb F$-stopping time belongs to $\mathcal H$ (we even have
 $Z_{\tau}I_{\{\tau<+\infty\}}=0$), and hence our subclass of honest times is not empty.\\
2) In the case where $\mathbb F$ is the completed Brownian
filtration, we consider the following $\mathbb F$-stopping times
$$
U^{\epsilon}_0=V^{\epsilon}_0=0,\ \ U^{\epsilon}_n:=\inf\{t\geq V^{\epsilon}_{n-1}:\ B_t
=\epsilon\},\ V^{\epsilon}_n:=\inf\{t\geq U^{\epsilon}_n:\ B_t=0\},$$ where $\epsilon\in (0,1)$
 and $B$ is a one dimensional standard Brownian motion. Then,
$$
\tau:=\sup\{V^{\epsilon}_n:\ V^{\epsilon}_n\leq T_1\},$$ where
$T_1:=\inf\{t\geq 0:\ \ B_t=1\}$, is an honest time which is not a
stopping time, and belongs to $\mathcal H$   (see
\cite{aksamit} for detailed proof). Other examples of elements of $\mathcal H$ that are not stopping times are given in the
next subsection.
\end{remark}

\noindent We conclude this subsection with the following lemma, obtained in
\cite{aksamit/choulli/deng/jeanblanc}.
\begin{lemma}\label{NUPBRforPredictableProcesses}
Let $X$ be an $\mathbb H$-predictable process with finite
variation. Then $X$ satisfies NUPBR$(\mathbb H)$ if and only if
$X\equiv X_0$ (i.e. the process $X$ is constant).\end{lemma}



\subsection{Particular Cases and Examples}
In this subsection, by analysing particular cases and examples, we
obtain some results vital for understanding the exact
interplay between the features of the initial markets and the
honest time under consideration. The following simple lemma
plays a key role in this analysis.

\begin{lemma}\label{localboundZ-} The following assertions hold.\\
{\rm{(a)}} Let $M$  be an ${\mathbb  F}$-local martingale,
and $\tau$ be an honest time. Then the process $\widehat M$,
defined as
\begin{eqnarray}\label{Mhat}\widehat{M}  &:=& M  -M^{\tau}+(1-Z_{-})^{-1}I_{\Lbrack\tau,+\infty\Rbrack}\is \langle M, \mm\rangle^{\mathbb F}\,,
 \end{eqnarray}
is a ${\mathbb  G }$-local martingale.\\
{\rm{(b)}} If $\tau \in \mathcal H$, then the $\mathbb
G$-predictable process
$\left(1-Z_{-}\right)^{-1}I_{\Lbrack\tau,+\infty\Rbrack}$ is
$\mathbb G$-locally bounded.
\end{lemma}

\begin{proof} {\bf 1)} Assertion (a)  is a standard  result on progressive enlargement of
filtration with honest times (see \cite{Barlow,DMM,Jeu}).\\
{\bf 2)} Herein we prove assertion (b). It is known \cite[Chapter XX]{DMM} that $ Z=\widetilde Z$  on
$\Lbrack\tau,+\infty\Rbrack$, and
$$\Lbrack\tau,+\infty\Rbrack\subset \{Z_{-}<1\}\cap\{\widetilde
Z<1\}\subset\{Z_{-}<1\}\cap\{ Z<1\}  \  .$$ Then, since   $\tau \in
\mathcal H$, we deduce that $\Rbrack\tau,+\infty\Rbrack\subset\{
Z<1\}$, and hence the process
$$X:=(1-Z)^{-1}I_{\Rbrack\tau,+\infty\Rbrack},$$ is c\`adl\`ag  $\mathbb G$-adapted with values in $[0,+\infty)$ (finite values).
 Combining these with $\Lbrack\tau,+\infty\Rbrack\subset \{Z_{-}<1\}$, we can prove easily that
$$T_n:=\inf\{t\geq 0:\ \ \ X_t\geq n\}\uparrow\ +\infty\ \mbox{and}\ \ \max(X^{T_n-}, X_{-}^{T_n})\leq n,\ \ \ P-a.s..$$
 Thus, $X_{-}=(1-Z_{-})^{-1}I_{\Lbrack\tau,+\infty\Rbrack}$ is locally bounded, and the proof of the lemma is completed.\qed
\end{proof}

\begin{theorem}\label{NUPR4Scontinuous} Suppose that $\tau \in \mathcal H$.
 If $S$ is continuous and satisfies  NUPBR$(\mathbb F)$, then $S-S^{\tau}$ satisfies  NUPBR$(\mathbb G)$.
\end{theorem}

\begin{remark}
This theorem follows from one of our principal results stated in the next subsection.
 However, due to the simplicity of its proof that does not require any further technicalities, we opted for detailing this proof below.
\end{remark}

\begin{proof}{\it of Theorem \ref{NUPR4Scontinuous}}:
Let $S=(S^1,...,S^d)$ be a $d$-dimensional continuous process
satisfying NUPBR$(\mathbb F)$. Then, there exists a positive
$\mathbb F$-local martingale $L$ such that $LS$ is an $(\mathbb
F,\sigma)$-martingale. Since $S$ is continuous and $L$ is a local
martingale, we deduce that $\sup_{u\leq .}\vert
S_u\vert\sup_{u\leq .}\vert \Delta L_u\vert$ is locally
integrable. Thus, thanks to Proposition 3.3 in
\cite{anselstricker1994} and $ \sum_{i=1}^d \Delta
(LS^i)=\sum_{i=1}^d S^i\Delta L\geq -d\sup_{u\leq .}\vert
S_u\vert\sup_{u\leq .}\vert \Delta L_u\vert$, we conclude that
$LS$ is an $\mathbb F$-local martingale. Consider a
sequence of $\mathbb F$-stopping times $(T_n)_{n\geq 1}$ that
increases to infinity such that both $L^{T_n}$ and
$L^{T_n}S^{T_n}$ are martingales, and put $Q_n:=(L_{T_n}/L_0)\is
P\sim P$. Then, $S^{(n)}:=S^{T_n}$ is an $(\mathbb F,
Q_n)$-martingale on the one hand. On the other hand, in virtue of
Proposition \ref{NUPBRLocalization}, $S-S^{\tau}$ satisfies
NUPBR$(\mathbb G)$ if and only if $S^{(n)}-(S^{(n)})^{\tau}$
satisfies NUPBR$(\mathbb G)$ under $Q_n$, for all $n\geq 1$. This
shows that, without loss of generality, one need to prove the
theorem only when $S$ is an $\mathbb F$-martingale. Thus,
for the rest of the proof, we assume that $S$ is an $\mathbb
F$-martingale. Thanks to Lemma \ref{localboundZ-}, the process
$Y^{\mathbb G}:={\cal
E}((1-Z_{-})^{-1}I_{\Lbrack\tau,+\infty\Rbrack}\is \widehat
{m^c})$ is a well defined continuous real-valued and positive
$\mathbb G$-local martingale, where $m^c$ is the continuous
$\mathbb F$-local martingale part of $m$, and $\widehat {m^c}$ is
defined as in (\ref{Mhat}). Thanks to the continuity of $S$ and
(\ref{Mhat}), we get
\begin{eqnarray*}
S-S^{\tau}+\Bigl[S-S^{\tau}, {{I_{\Lbrack\tau,+\infty\Rbrack}}\over{1-Z_{-}}}\is \widehat {m^c}\Bigr]
&=&S-S^{\tau}+(1-Z_{-})^{-1}I_{\Lbrack\tau,+\infty\Rbrack}\is\langle S, m\rangle^{\mathbb F}\\
\\
&=&\widehat S\in {\cal M}_{loc}(\mathbb G).\end{eqnarray*}
Therefore, a combination of this and It\^o's formula applied to
$(S-S^{\tau})Y^{\mathbb G}$, we conclude that this latter process
is a $\mathbb G$-local martingale. This proves the NUPBR$(\mathbb G)$
for $S-S^{\tau}$, and the proof of the theorem is achieved.\qed
\end{proof}

\begin{remark}\label{remarkContinuity}
Theorem \ref{NUPR4Scontinuous} asserts clearly that, if $\tau \in \mathcal H$,
the jumps of $S$ have significant impact on $\mathbb G$-arbitrages
for $S-S^{\tau}$. Thus, the following natural question arises:
\begin{equation}\label{questioncontinuity}
\mbox{Does the condition}\ \{\Delta S\not=0\}\cap\Rbrack\tau\Lbrack=\emptyset\  \mbox{impact}\ \mathbb G\mbox{-arbitrages}?
\end{equation}
\end{remark}

\begin{example}\label{example1} Suppose that $\mathbb F$ is generated by a Poisson process $N$ with intensity one.
 Consider two real numbers $a>0$ and $\mu>1$, and set
\begin{equation}\label{tau}
\tau:=\sup\{t\geq 0:\ Y_t:=\mu t-N_t\leq a\},\ \ \ \ \ \
M_t:=N_t-t.\end{equation} It can be proved easily, see
\cite{aksamit}, that $\tau\in \mathcal H$ is finite almost surely,
and the associated processes $Z$ and $\widetilde Z$ are given by
$$
Z=\Psi(Y-a)I_{\{ Y\geq a\}}+I_{\{Y<a\}}\ \ \ \ \mbox{and}\ \ \ \ \widetilde Z=\Psi(Y-a)I_{\{Y>a\}}+I_{\{Y\leq a\}}.$$
Here $\Psi(u):=P\left(\sup_{t\geq 0}Y_t>u\right)$ is the ruin probability associated to the process $Y$ (see \cite{assmussen}).
 As a result we have
\begin{equation}\label{Z-}
1-Z_{-}=\left[1-\Psi(Y_{-}-a)\right]I_{\{ Y_{-}>a\}},
\end{equation}
and we can prove that
 \begin{eqnarray}\label{mrepresentation}
m&=&m_0+\phi\is M,\ \ \ \mbox{where}\\
\phi:&=&\left[\Psi(Y_{-}-a-1)-\Psi(Y_{-}-a)\right]I_{\{ Y_{-}>1+
a\}}+\left[1-\Psi(Y_{-}-a)\right]I_{\{ a<Y_{-}\leq
1+a\}}.\nonumber\end{eqnarray} Suppose that $S= I_{\{ a\leq
Y_{-}<a+  1\}}\is M$. Then, in virtue of Lemma
\ref{NUPBRforPredictableProcesses}, the process $S-S^{\tau}$
(which is not null) violates  NUPBR$(\mathbb G)$ if it is $\mathbb
G$-predictable with finite variation. This latter fact is
equivalent to $\widehat S$ ($\mathbb G$-local martingale part of
$S-S^{\tau}$) being null, or equivalently $\langle \widehat
S,\widehat S\rangle^{\mathbb G}\equiv 0$. By using Lemma
\ref{localboundZ-} and It\^o's lemma and putting $V_t=t$, we
derive
\begin{eqnarray}\label{ShatS}
[\widehat S,\widehat S]&=&I_{\Lbrack\tau,+\infty\Rbrack}\is [S]=I_{\Lbrack\tau,+\infty\Rbrack}\is S+I_{\{a< Y_{-}\leq a+1\}}I_{\Lbrack\tau,+\infty\Rbrack}\is V\nonumber\\
&=&I_{\Lbrack\tau,+\infty\Rbrack}\is {\widehat S}+I_{\{a< Y_{-}\leq a+1\}}I_{\Lbrack\tau,+\infty\Rbrack}\left(1-{{\phi}\over{1-Z_{-}}}\right)\is V,\\
&=&I_{\Lbrack\tau,+\infty\Rbrack}\is {\widehat S}\ \ \ \mbox{is a $\mathbb G$-local martingale}.\nonumber
\end{eqnarray} The last equality is due to $\phi\equiv 1-Z_{-}$ on $\{a\leq Y_{-}<a+1\}\cap\Lbrack\tau,+\infty\Rbrack$. This proves that $\widehat S\equiv 0$, and hence $S-S^{\tau}$ violates  NUPBR$(\mathbb G)$.
\end{example}

\begin{example}\label{example2} Consider the same setting and notations as Example \ref{example1},
 except for the initial market model that we suppose having the form of $S=I_{\{ Y_{-}>a+1\}}\is M$ instead.
  Then, by combining Lemma \ref{localboundZ-}, It\^o's lemma and similar calculation as in (\ref{ShatS}),
   we deduce that both $Y^{\mathbb G}:={\cal E}(\xi\is \widehat S)$ and $Y^{\mathbb G}(S-S^{\tau})$
    are $\mathbb G$-local martingales and $Y^{\mathbb G}>0$. Here $\xi$ is given by
$$
\xi:={{\Psi(Y_{-}-a-1)-1}\over{2-\Psi(Y_{-}-a)-\Psi(Y_{-}-a-1)}}I_{\{ Y_{-}>a+1\}}I_{\Lbrack\tau,+\infty\Rbrack}.$$
 This proves that $S-S^{\tau}$ satisfies  NUPBR$(\mathbb G)$.
\end{example}

\begin{remark}\label{norolefortau}
1) The economics/financial meaning of Examples \ref{example1} and \ref{example2} resides in the following:
 The random time defined in (\ref{tau}) represents the last time the cash reserve of a firm does not exceed the level $a$.
  Then, in  Example \ref{example1} (respectively in Example \ref{example2}) one can consider any security whose price process
   lives on $\{a\leq Y_{-}<1+a\}$ (respectively on  $\{Y_{-}> 1+a\}$).\\
2) Remark that, in both Examples \ref{example1}
and \ref{example2}, the graph of the random time $\tau$ is
included in a union of countable graphs of predictable stopping
times. Hence, due to the quasi-left-continuity of $S$, we
immediately conclude that $\{\Delta
S\not=0\}\cap\Rbrack\tau\Lbrack$ is empty for both examples. This answers negatively (\ref{questioncontinuity}).
\end{remark}

\subsection{Main Results and Their Applications}

Our first main result requires the following easy and interesting lemma.

\begin{lemma}\label{VFprocess}
Suppose that $\tau\in {\cal H}$ and is finite almost surely. Then,
\begin{equation}\label{VFdefinition}
V^{\mathbb F}:=\sum I_{\{\widetilde Z=1>Z_{-}\}},
\end{equation}
is c\`adl\`ag with finite values, and hence is $\mathbb F$-locally integrable.
\end{lemma}
\begin{proof} Thanks to Proposition \ref{lemma:predsetFG1}-(b), there exists a sequence of $\mathbb F$-stopping times,
 $(\sigma_n)_{n\geq 1}$ that increases to infinity almost surely and  $1\leq n^2(1-Z_{t-})^2$ on $\{Z_{t-}<1\}\cap\{t\leq\sigma_n\}$.
  Thus, for any nonnegative and bounded $\mathbb F$-optional process $H$, we have
\begin{eqnarray*}
\left(H\is V^{\mathbb F}\right)^{\sigma_n}&\leq& n^2\left(\sum H(1-Z_{-})^2I_{\{\widetilde Z=1>Z_{-}\}}\right)^{\sigma_n}\\
&=& n^2\left(\sum H(\Delta m)^2I_{\{\widetilde Z=1>Z_{-}\}}\right)^{\sigma_n}\leq n(H\is [m,m])^{\sigma_n}.
\end{eqnarray*}
Therefore, the proof of the lemma follows immediately from combining the above inequality and
 the fact that $ [m,m]\in {\cal A}^+_{loc}(\mathbb F)$.\qed\end{proof}

\noindent In the following, we announce our first main result.

\begin{theorem}\label{individualSaftertau0} Suppose that $S$ is an $\mathbb F$-quasi-left-continuous semimartingale, and $\tau \in \mathcal H$
is finite almost surely. Then, the following are equivalent.\\
(a) $S-S^{\tau}$ satisfies  NUPBR$(\mathbb
G)$.\\
(b) $(1-Z_{-})\is {\cal T}_a(S)$ satisfies NUPBR$(\mathbb F)$.\\
(c) $I_{\{Z_{-}<1\}}\is {\cal T}_a(S)$ satisfies NUPBR$(\mathbb F)$, where
\begin{equation}\label{assertion1}
{\cal T}_a(S):=S-[S,V^{\mathbb F}]=S-\sum \Delta S I_{\{\widetilde Z=1>Z_{-}\}}.\end{equation}
\end{theorem}

\noindent The proof of (a)$\Longrightarrow$(b) is technical and requires notations. Thus, for the reader's convenience, we
postponed the whole theorem's proof to Section \ref{Sectionproofs}.

\begin{remark}
(a) The theorem asserts, in a precise and deep manner, that $\mathbb G$-arbitrages for the process $S-S^{\tau}$
 are intimately related to the interplay between the jumps of $S$ and the jumps of $\widetilde Z$ to the value one.\\
(b) Theorem \ref{individualSaftertau0} claims that $S-S^{\tau}$ is arbitrage-free under
 $\mathbb G$ if and only if the part ${\cal T}_a(S)$ (of $S$) is arbitrage-free under
  $\mathbb F$ on the set $\{Z_{-}<1\}$. As a result, this allows us to single out practical
   cases for which the NUPBR is preserved after $\tau$, as
outlined in the forthcoming Corollary \ref{corollaryofmain1} and
Theorem \ref{BKPredictableJumpsaftertau3}.
\end{remark}

\begin{corollary}\label{corollaryofmain1}
Suppose that $S$ is $\mathbb F$-quasi-left-continuous, and $\tau\in  \mathcal H$ is finite almost surely. Then the following assertions hold:\\
 {\rm{(a)}} If $\left(S, \sum (\Delta S)I_{\{\widetilde Z=1>Z_{-}\}} \right)$ satisfies  NUPBR$(\mathbb F)$,
  then $S-S^{\tau}$ satisfies  NUPBR$(\mathbb G)$.\\
{\rm{(b)}} If  $S$ satisfies   NUPBR$(\mathbb F)$ and $\{\Delta S\not=0\}\cap\{\widetilde
Z=1>Z_{-}\}=\emptyset$, then $S-S^{\tau}$
satisfies   NUPBR$(\mathbb G)$.
\end{corollary}

\begin{remark} 1) Assertion (b) asserts that if $S$ does not jump on $\{\widetilde Z=1>Z_{-}\}$,
 then no arbitrage under $\mathbb G$ will occur in the part ``after-$\tau$". Assertion (a) gives
  much weaker assumption than assertion (b), as it assumes that
   $L\left(\sum \Delta SI_{\{\widetilde Z=1>Z_{-}\}}\right)\in {\cal M}_{loc}(\mathbb F)$ for some $L\in {\cal L}_{\sigma}(S,\mathbb F)$
    (defined in (\ref{Set4SigmaMgdensities})),
    while assertion (b) assumes that $\sum \Delta SI_{\{\widetilde Z=1>Z_{-}\}}$ is null. \end{remark}

\begin{proof}{\it of Corollary \ref{corollaryofmain1}:} It is obvious that assertion (a) follows directly from combining
 $(1-Z_{-})\is {\cal T}_a(S)=\left( 1-Z_{-},-(1-Z_{-})\right)\is \left(S,  \sum \Delta SI_{\{\widetilde Z=1>Z_{-}\}}\right)$
  and Theorem \ref{individualSaftertau0}. Due to $\{\Delta S\not=0\}\cap\{\widetilde
Z=1>Z_{-}\}=\emptyset$, assertion (b) follows from assertion (a), and the proof of the corollary is achieved.\qed\end{proof}

\noindent In the spirit of further applicability of Theorem \ref{individualSaftertau0}, we state the following

\begin{theorem}\label{BKPredictableJumpsaftertau3} Suppose that $\tau \in \mathcal H$.
Let $\mu$ be the optional random measure associated to the jumps
of $S$, and $\nu^{\mathbb F}$ and $\nu^{\mathbb G}$ be the $\mathbb F$-compensator and the $\mathbb G$--compensator of
$\mu$ and $I_{\Lbrack\tau,+\infty\Rbrack}\cdot\mu$ respectively. If $S$
satisfies NUPBR$(\mathbb F)$ and
\begin{equation}\label{LevyAssumption}
I_{\Lbrack\tau,+\infty\Rbrack}\cdot \nu^{\mathbb F}\ \ \mbox{is equivalent to }\ \ \nu^{\mathbb G}\ \ \ \ \ \ \ \ \ P-a.s.,
\end{equation}
then $S-S^{\tau}$ satisfies NUPBR$(\mathbb G)$.
\end{theorem}

\noindent The proof of this theorem follows from Theorem \ref{individualSaftertau0} as long as we prove that,
 under (\ref{LevyAssumption}), $S$ satisfies the NUPBR$(\mathbb F)$ if and only if ${\cal T}_a(S)$
  satisfies the NUPBR$(\mathbb F)$. This proof is technical, and thus it is delegated to Section \ref{Sectionproofs}.

\begin{remark} Remark that we always have the absolute continuity
 $\nu^{\mathbb G}\ll I_{\Lbrack\tau,+\infty\Rbrack}\cdot \nu^{\mathbb F}$ $P-a.s.$
  This follows from the fact that $\nu^{\mathbb G}$ is absolutely continuous with respect to
   $\nu^{\mathbb F}$ and it lives on $\Lbrack\tau,+\infty\Rbrack$ only. \\
(a) {\bf The L\'evy Case:} Suppose that $S$ is a L\'evy process and $F(dx)$ is its L\'evy measure under $\mathbb F$,
 then $\nu^{\mathbb F}(dt,dx)=F(dx)dt$ and $\nu^{\mathbb G}(dt,dx)=I_{\Lbrack\tau,+\infty\Rbrack}F^{\mathbb G}_t(dx)dt$,
  where $F^{\mathbb G}_t(dx)$ is its L\'evy measure under $\mathbb G$. Thus, Theorem \ref{BKPredictableJumpsaftertau3}
   asserts that if $P\otimes \lambda$ almost every $(\omega, t)$ ($\lambda(dt)=dt)$, $F^{\mathbb G}_t(\omega, dx)=f(t,x,\omega)F(dx)$
    for some real-valued and positive functional $f(t,x,\omega)$, then $S-S^{\tau}$ satisfies  NUPBR$(\mathbb G)$.
     For more practical L\'evy cases, we refer the reader to \cite{Deng2014}.\\
(b) {\bf Examples \ref{example1}--\ref{example2} versus Theorem
\ref{BKPredictableJumpsaftertau3}:} In the context of Example
\ref{example1}, we easily calculate $\nu^{\mathbb F}(dt,dx)=I_{\{
a<Y_{t-}\leq a+1\}}\delta_1(dx)dt$ and $\nu^{\mathbb
G}(dt,dx)=I_{\Lbrack\tau,+\infty\Rbrack}(t)I_{\{ a<Y_{t-}\leq
a+1\}}\left(1-\phi_t/(1-Z_{t-})\right)\delta_1(dx)dt\equiv 0$
which is not equivalent to
$I_{\Lbrack\tau,+\infty\Rbrack}\cdot\nu^{\mathbb F}$. This example
shows that (\ref{LevyAssumption}) can be violated. Therefore, in
those circumstances, we can not conclude whether $S-S^{\tau}$
satisfies  NUPBR$(\mathbb G)$ or not directly from Theorem
\ref{BKPredictableJumpsaftertau3}.
\\
For the case of Example \ref{example2}, we have $\nu^{\mathbb F}(dt,dx)=I_{\{ Y_{t-}>a+1\}}\delta_1(dx)dt$
 and $\nu^{\mathbb G}(dt,dx)=I_{\Lbrack\tau,+\infty\Rbrack}(t)I_{\{ Y_{t-}> a+1\}}\left(1-\phi_t/(1-Z_{t-})\right)\delta_1(dx)dt$
  which is equivalent to $I_{\Lbrack\tau,+\infty\Rbrack}\cdot\nu^{\mathbb F}$ since $\{Y_{-}>a+1\}\subset\{\phi<1-Z_{-}\}$ $P\otimes dt$-a.e.
   Thus, Theorem \ref{BKPredictableJumpsaftertau3} allows us to conclude that $S-S^{\tau}$ fulfills the NUPBR$(\mathbb G)$.
\end{remark}

\noindent The rest of this subsection describes models of $\tau$ preserving the NUPBR.

\begin{theorem}\label{CorollaryBKapresdefault} Assume that $\tau \in {\mathcal H}$.
 Then, the following are equivalent.\\
{\rm{(a)}} The set $\{\widetilde Z=1>Z_{-}\}$ is accessible (i.e. it is
contained in a countable union of graphs of $\mathbb F$-predictable stopping times).\\
{\rm{(b)}} For every (bounded) $\mathbb F$-quasi-left-continuous
martingale  $X$, the process  $X-X^{\tau}$ satisfies NUPBR$(\mathbb G)$.\\
{\rm{(b')}} For any probability $Q\sim P$ and every (bounded) $\mathbb F$-quasi-left-continuous
$X\in {\cal M}(Q,\mathbb F)$, the process  $X-X^{\tau}$ satisfies NUPBR$(\mathbb G)$.\\
{\rm{(c)}} For every (bounded) $\mathbb F$-quasi-left-continuous process $X$ satisfying NUPBR$(\mathbb F)$,
 the process  $X-X^{\tau}$ satisfies NUPBR$(\mathbb G)$.
 \end{theorem}

\begin{proof} The proof of the proposition is organized in three parts, where we prove (a)$\Longleftrightarrow$(b),
 (b)$\Longleftrightarrow$(b') and (b')$\Longleftrightarrow$(c) respectively.\\
{\bf 1)} We start by proving that (a)$\Rightarrow $(b).
Suppose that the thin set $\{\widetilde Z=1>Z_{-}\}$ is
accessible. Then, for any $\mathbb F$-quasi-left-continuous martingale $X$, we have $\{\Delta X\not=0\}\cap\{\widetilde
Z=1>Z_{-}\}=\emptyset$. Hence, thanks to Corollary \ref{corollaryofmain1}--(d), we deduce that $X-X^{\tau}$ satisfies NUPBR$(\mathbb
G)$. This completes the proof of (a)$\Rightarrow $(b). To prove the reverse,  assuming that assertion (b) holds,  we
consider a sequence of stopping times $(T_n)_{n\geq 1}$ that
exhausts the thin set $\{\widetilde Z=1>Z_{-}\}$ (i.e.,
$\left\{\widetilde Z=1>Z_{-}\right\}=\displaystyle\bigcup_{n=1}^{+\infty}\Rbrack
T_n\Lbrack$). Then, each $T_n$ -- that we denote by $T$ for the
sake of simplicity--  can be decomposed into a totally
inaccessible part $T^i$ and   an accessible part $T^a$ as
$T=T^i\wedge T^a$. Consider the following quasi-left-continuous $\mathbb F$-martingale
$$
M:=V-V^{p,\mathbb F}=:V-\widetilde V,$$
where $V:=I_{\Rbrack T^i,+\infty\Rbrack}$. Then, since   $\{
T^i<+\infty\}\subset\{\widetilde Z_{T^i}=1\}$, we deduce that
$\{T^i<+\infty\}\subset\{\tau\geq T^i\}$ and hence
$$
I_{\Lbrack \tau,+\infty\Rbrack}\is M=-I_{\Lbrack
\tau,+\infty\Rbrack}\is\widetilde V\ \ \ \ \mbox{is $\mathbb
G$-predictable}.$$ Then, the finite variation and $\mathbb G$-predictable process, $I_{\Lbrack \tau,+\infty\Rbrack}\is M$, satisfies NUPBR$(\mathbb
G)$ if and only if it is null, or equivalently
$$
0=E\left(I_{\Lbrack \tau,+\infty\Rbrack}\is\widetilde V_{\infty}\right)=E\left(\int_0^{\infty}(1-Z_{s-})d\widetilde V_s\right)=
E\left((1-Z_{T^i-})I_{\{ T^i<+\infty\}}\right).$$ Therefore, we conclude that $T^i=+\infty,\ P-a.s.$, and the stopping time $T$ is an accessible stopping time. This ends the proof of (a)$\Longleftrightarrow$ (b).\\

\noindent {\bf 2)} It is easy to see that the implication
(b')$\Longrightarrow$ (b) follows from taking $Q=P$. To prove the
reverse sense, we suppose given $Q\sim P$ and an $\mathbb
F$-quasi-left-continuous $X\in {\cal M}(\mathbb F,Q)$. Then, put
$$Z^{\mathbb F}_t:=E\Bigl({{dQ}\over{dP}}|{\cal F}_t\Bigr)=:{\cal E}_t(N),\ \ \ Y:=\left(\begin{array}{ll}
{\cal E}(N^{(qc)})X\\
{\cal E}(N^{(qc)})\end{array}\right)\ \ \mbox{and}\
N^{(qc)}:=N-I_{\bigcup_n\Rbrack \sigma_n\Lbrack}\is N,$$ where
$(\sigma_n)_n$ is the sequence of $\mathbb F$-predictable stopping
times that exhausts all the predictable jumps of $N$. In other
words, $N^{(qc)}$ is the $\mathbb F$-quasi-left-continuous local
martingale part of $N$.  Then, due to the quasi-left-continuity of
$X$, simple calculations show that $Y$ is an $\mathbb
F$-quasi-left-continuous martingale. Therefore, by a directly
applying assertion (b) to $Y$, we conclude that
$Y-Y^{\tau}=\left(\begin{array}{ll}
{\cal E}(N^{(qc)})(X-X^{\tau})+X^{\tau}({\cal E}(N^{(qc)})-{\cal E}(N^{(qc)})^{\tau})\\
{\cal E}(N^{(qc)})-{\cal E}(N^{(qc)})^{\tau}\end{array}\right)$ satisfies  NUPBR$(\mathbb G)$. This implies the existence of a real-valued positive $\mathbb G$-local martingale $Z^{\mathbb G}$ such that both processes $Z^{\mathbb G}{\cal E}(N^{(qc)})(X-X^{\tau})$ and $Z^{\mathbb G}{\cal E}(N^{(qc)})$ are $\sigma$-martingales under $(\mathbb G,P)$. Since $Z^{\mathbb G}{\cal E}(N^{(qc)})$ is positive and thanks to Proposition 3.3 and Corollary 3.5 of \cite{anselstricker1994} (which states that a non-negative $\sigma$-martingale is a local martingale), we deduce that  $Z^{\mathbb G}{\cal E}(N^{(qc)})$ is a real-valued positive element of ${\cal M}_{loc}(\mathbb G,P)$ such that  $Z^{\mathbb G}{\cal E}(N^{(qc)})(X-X^{\tau})$ is a $\sigma$-martingale. This proves that $X-X^{\tau}$ satisfies  NUPBR$(\mathbb G)$, and the proof of (b)$\Longleftrightarrow$ (b') is completed.\\

\noindent{\bf 3)} Remark that (c) $\Longrightarrow $ (b') is obvious, and hence we focus on proving the reverse only. Suppose that assertion (b') holds, and consider an $\mathbb F$-quasi-left-continuous process $X$ satisfying NUPBR$(\mathbb F)$. Then, there exists a real-valued and positive $\mathbb F$-local martingale $Y$, and a real-valued and $\mathbb F$-predictable process $\phi$ such that
$$
0<\phi\leq 1\ \ \ \ \ Y(\phi\is X)\ \ \ \mbox{is an $\mathbb F$-martingale}.$$
Let $(T_n)$ be a sequence of $\mathbb F$-stopping times that increases to infinity (almost surely) such that $Y^{T_n}$ is a martingale, and set
 $$
 \overline X:=\phi\is X,\ \ \ \ Q_n:=Y_{T_n}/Y_0\is P\sim  P.$$
 By applying assertion (b') to  $\overline{X}^{T_n}$  and $Q_n\sim P$ (since $\overline{X}^{T_n}$ is an $\mathbb F$-quasi-left-continuous element of ${\cal M}(\mathbb F,Q_n)$), we conclude that $\left(\phi\is (X-X^{\tau})\right)^{T_n}=\overline{X}^{T_n}-(\overline{X}^{T_n})^{\tau}$ satisfies NUPBR$(\mathbb G)$. Hence, thanks -again- to Proposition \ref{NUPBRLocalization},  NUPBR$(\mathbb G)$ for $X-X^{\tau}$ follows immediately. This ends the proof of (b) $\Longleftrightarrow$ (c), and that of the proposition as well.\qed
\end{proof}

\begin{theorem}\label{F-quasi-left-continuous}
Suppose that $\tau \in {\mathcal H}$  and $\mathbb F$ is quasi-left-continuous. Then the following assertions are equivalent.\\
{\rm{(a)}} The thin set $\{\widetilde Z=1>Z_{-}\}$ is evanescent.\\
{\rm{(b)}} For every (bounded) $\mathbb F$-martingale  $X$, the process  $X-X^{\tau}$ satisfies NUPBR$(\mathbb G)$.\\
{\rm{(b')}} For any probability $Q\sim P$ and every (bounded) $\mathbb F$-quasi-left-continuous $X\in {\cal M}(Q,\mathbb F)$, the process  $X-X^{\tau}$ satisfies NUPBR$(\mathbb G)$.\\
{\rm{(c)}} For every (bounded) $X$ satisfying NUPBR$(\mathbb F)$, $X-X^{\tau}$ satisfies the NUPBR$(\mathbb G)$.
\end{theorem}
\begin{proof} The proofs of both equivalences (b') $\Longleftrightarrow $(c) and  (b) $\Longleftrightarrow $(b') follow the same arguments as the corresponding proofs in Theorem \ref{CorollaryBKapresdefault} (see parts 2) and 3)). Hence, we omit these proofs and the proof of (a) $\Longrightarrow $(b) as well, as this latter one follows immediately from Theorem \ref{CorollaryBKapresdefault}-(a) or Corollary \ref{corollaryofmain1}--(d). Thus, the remaining part of the proof focuses on proving (a) $\Longrightarrow $(b). To this end, we assume that assertion (b) holds, and recall that --when $\mathbb F$ is a quasi-left-continuous filtration-- any accessible $\mathbb F$-stopping time is predictable (see \cite{Dellacherie72} or  \cite[Th. 4.26]{Yanbook}). Then, since $\mathbb F$ is a quasi-left-continuous filtration, any $\mathbb F$-martingale is quasi-left-continuous, and from Theorem \ref{CorollaryBKapresdefault} we deduce that the thin set, $\{\widetilde Z=1<Z_{-}\}$, is predictable. Now take any $\mathbb F$-predictable stopping time $T$ such that
$$
\Rbrack T\Lbrack\subset \{\widetilde Z=1>Z_{-}\}.$$
This implies that $\{T<+\infty\}\subset\{\widetilde Z_T=1\}$, and due to $E(\widetilde Z_T|{\cal F}_{T-})=Z_{T-}$ on $\{T<+\infty\}$, we get
$$E(I_{\{T<+\infty\}}(1-Z_{T-}))=E(I_{\{T<+\infty\}}(1-\widetilde{Z}_{T}))=0.$$

\noindent This leads to $T=+\infty\ P-a.s$ (since $\{T<+\infty\}\subset\{Z_{_T-}<1\}$), and the proof of the theorem is completed.\qed\end{proof}

\begin{remark}
The conclusion of Theorem \ref{F-quasi-left-continuous} remains valid without
 the quasi-left-continuous assumption on the filtration $\mathbb F$. This general case,
  that can be found in the earlier version \cite{aksamit/choulli/deng/jeanblanc},
  requires more technical arguments.\end{remark}


\begin{remark}\label{argum4sectiondeflator}
 The proof of (b) $\Longleftrightarrow$ (c) in Theorem \ref{individualSaftertau0} is obvious
 (due to the fact that $(1-Z_{-})^{-1}I_{\{Z_{-}<1\}}$ is $\mathbb F$-locally bounded --see Proposition \ref{lemma:predsetFG1}--).
 Thus, the only parts that require proof are (a) $\Longleftrightarrow$(b). The implication
(b) $\Longrightarrow$(a) can be formulated in a more abstract way, due to the simple fact that
 ${\cal T}_a(S)-({\cal T}_a(S))^{\tau}=S-S^{\tau}$  and
${\cal T}_a(S)$ does not jump on $\{\widetilde Z=1>Z_{-}\}$. Thus, in virtue of Proposition \ref{NUPBRLocalization}, one can assume without
loss of generality that $S$ is an $\mathbb F$-quasi-left-continuous local martingale that does not jump on $\{\widetilde Z=1>Z_{-}\}$,
and prove that in this case $S-S^{\tau}$ satisfies the NUPBR$(\mathbb G)$. This is the aim of the next section in a more
interesting and general manner, as it constructs explicitly a $\mathbb G$-deflator for any $S-S^{\tau}$
as long as $S\in{\cal M}_{loc}(\mathbb F)$, quasi-left-continuous, and orthogonal to
$V^{\mathbb F}-\left(V^{\mathbb F}\right)^{p,\mathbb F}$ ($V^{\mathbb F}$ is defined in (\ref{VFdefinition})).\end{remark}
\section{Explicit Deflators for a Class of $\mathbb F$-Local Martingales}\label{explicitdeflators}
 This section proposes an explicit construction of $\mathbb G$-deflators for $M-M^{\tau}$,
  when $M$ spans a class of $\mathbb F$-quasi-left-continuous local martingales.
   The key mathematical idea behind this achievement lies in the exact relationship between the $\mathbb G$-compensator
    and the $\mathbb F$-compensator of a process with finite variation when both exists.
 This is the aim of the first subsection, while the second subsection states the results about deflators.
 \subsection{Dual Predictable Projections under $\mathbb G$ and $\mathbb F$}\label{Stochastic1}

\noindent In the following, we start our study  by writing the
$\mathbb G$-compensators/projections in terms of $\mathbb
F$-compensators/projections respectively. Even though, the proofs of the results of this subsection are easy and not technical at all,
 we opted for delegating them to the Appendix for the reader's convenience.

\begin{lemma}\label{lemmecrucialapresdefault} Suppose that  $\tau \in {\mathcal H}$. Then, the following assertions hold.\\
{\rm{(a)}}  For any $\mathbb F$-adapted process  $V$, with locally
integrable variation we have
\begin{eqnarray}\label{Gcompensatoraftertau}
  I_{\Lbrack\tau,+\infty\Rbrack}\is   V ^{p,\mathbb G} &=&
   I_{\Lbrack\tau,+\infty\Rbrack}(1-Z_{-})^{-1}\is \left((1-\widetilde Z) \is V\right)^{p,\mathbb F} \, ,
 \end{eqnarray}
 and  on $\Lbrack\tau,+\infty\Rbrack$ 
 \begin{equation}\label{predictableprojectionsjumpsaftertau}
 ^{p,{\mathbb G}}\left({\Delta V}\right) =  (1-Z_{-})^{-1}\ ^{p,{\mathbb F}}\left( (1-\widetilde Z){\Delta V} \right)\,.
 \end{equation}
{\rm{(b)}} For any $\mathbb F$-local martingale $M$,   one has, on $\Lbrack\tau,+\infty\Rbrack$
\begin{equation}\label{inclusionofsetsaftertau}
^{p,{\mathbb G}}\left({{\Delta M}\over{ 1-\widetilde Z}}\right) =
{{^{p,{\mathbb F}}\left({\Delta M}I_{\{\widetilde Z<1\}}\right)}\over{1-Z_{-}}},\ \ \mbox{and}\ \
^{p,{\mathbb G}}\left({1\over{ 1-\widetilde Z}}\right)={{^{p,{\mathbb F}}\left(I_{\{\widetilde Z<1\}}\right)}\over{1-Z_{-}}}.\end{equation}
{\rm{(c)}} For any quasi-left-continuous $\mathbb F$-local martingale $M$,  one has
 \begin{equation}\label{UsefulProjections}
  ^{p,{\mathbb G}}\left((\Delta M) (1-\widetilde Z)^{-1}I_{\Lbrack\tau,+\infty\Rbrack}\right)  = 0.
\end{equation}
 \end{lemma}

\noindent The next lemma focuses on the integrability of the process $ (1-\widetilde Z)^{-1}I_{\Lbrack\tau,+\infty\Rbrack}$
 with respect to any process with $\mathbb F$-locally integrable variation. As a result, we complete our comparison of
  $\mathbb G$ and $\mathbb F$ compensators. Recall that, due to \cite[Chapter XX]{DMM},  $\widetilde Z=Z$ on $\Lbrack\tau,+\infty\Rbrack$.

\begin{lemma}\label{maintechnique}
Let $\tau$ be an honest time and $V$ be a c\`adl\`ag and $\mathbb F$-adapted process with finite variation. Then, the following assertions hold.\\
(a) The process
\begin{equation}\label{ZcdotV} U:=\left(1-Z\right)^{-1}I_{\Lbrack
\tau,+\infty\Rbrack}\is V,\end{equation}
is a well defined process, that is $\mathbb G$-adapted, c\`adl\`ag and  has finite variation.\\
(b) If $V$ belongs to ${\cal A}_{loc}(\mathbb F)$ (respectively to ${\cal A}(\mathbb F)$), then $U\in {\cal A}_{loc}(\mathbb G)$
 (respectively $U\in {\cal A}(\mathbb G)$) and
\begin{equation}\label{GFcompensator1}
   U ^{p,\mathbb G}=I_{\Lbrack \tau,+\infty\Rbrack}(1-Z_{-})^{-1}\is(I_{\{ {\widetilde Z}<1\}}\is V)^{p,\mathbb F}.
\end{equation}
(c) Suppose furthermore that $\tau$ is finite almost surely. Then, $I_{\Lbrack \tau,+\infty\Rbrack}\is V\in {\cal A}_{loc}(\mathbb G)$
 if and only if $(1-\widetilde Z)\is V\in {\cal A}_{loc}(\mathbb F)$.\\
(d) Suppose furthermore that $\tau$ is finite almost surely, and $V$ is a nondecreasing and  $\mathbb
F$-predictable process. Then, for any $\mathbb
F$-predictable process $\varphi\geq0$,\\ $\varphi I_{\Lbrack
\tau,+\infty\Rbrack}\is V\in {\cal A}^+_{loc}(\mathbb G)$ iff $(1-Z_{-})\varphi\is V\in {\cal
A}^+_{loc}(\mathbb F)$ iff $\varphi I_{\{Z_{-}<1\}}\is V\in {\cal
A}^+_{loc}(\mathbb F)$.
\end{lemma}

 \subsection{Construction of Deflators}
 Herein, we start by introducing a deflator-candidate as follows.

\begin{proposition}\label{lemmforNtilde1}Suppose that $\tau\in {\mathcal H}$ and
consider the $\mathbb G$-local martingale
  \begin{equation}\label{processNNhat1}
  \widehat{m} :=  I_{\Lbrack \tau,+\infty\Rbrack}\is m + (1-Z_{-})^{-1}I_{\Lbrack \tau,+\infty\Rbrack}\is\langle m \rangle^{\mathbb F},\end{equation}
 and the process
 \begin{equation}\label{Kappa/WG}
  W^{\mathbb G}:=\Bigl((1-Z_{-})(1-\widetilde Z)\Bigr)^{-1}I_{\Lbrack\tau,+\infty\Rbrack}\is [m,m].
\end{equation}
Then, the following assertions hold.\\
1) The nondecreasing and $\mathbb G$-optional process $W^{\mathbb G}$ belongs to ${\cal A}^+_{loc}(\mathbb G)$.\\
2) The $\mathbb G$-local martingale
\begin{equation}\label{LG}
L^{\mathbb G}:=(1-Z_{-})^{-1}I_{\Lbrack \tau,+\infty\Rbrack}\is\widehat m+W^{\mathbb G}-\left(W^{\mathbb G}\right)^{p,\mathbb G},
\end{equation}
satisfies the following properties:\\
{\rm{(2-a)}} ${\cal E}( L^{\mathbb G})>0$ (or equivalently $1+\Delta  L^{\mathbb G}>0$) and  $I_{\Lbrack 0,\tau\Lbrack}\is  L^{\mathbb G}=0.$\\
{\rm{(2-b)}} For any  $M\in {\cal M}_{0,loc}(\mathbb F)$, we
have
\begin{equation}\label{localintegrablebracket1}
 [ L^{\mathbb G}, \widehat M ]\in{\cal A}_{loc}(\mathbb G)\ \
 \left(\mbox{i.e.}\ \langle  L^{\mathbb G},\widehat  M \rangle^{\mathbb G}\ \mbox{ exists}\right),\end{equation}
 where $\widehat M$ is defined in (\ref{Mhat}).
\end{proposition}

\begin{proof}Thanks to Lemma \ref{localboundZ-}-(b), $(1-Z_{-})^{-1}I_{\Lbrack\tau,+\infty\Rbrack}$ is $\mathbb G$-locally bounded.
 Thus, by combining this fact with $[m,m]\in {\cal A}^+_{loc}(\mathbb F)$ and Lemma \ref{maintechnique}-(b),
  we conclude that $W^{\mathbb G}=(1-Z_{-})^{-1}(1-\widetilde Z)^{-1}I_{\Lbrack\tau,+\infty\Rbrack}\is [m,m]\in {\cal A}^+_{loc}(\mathbb G)$,
   and subsequently assertion (1) holds. Thus,
   the process $L^{\mathbb G}$ --given in (\ref{LG})-- is a well defined $\mathbb G$-local martingale.
    The rest of this proof focuses on proving the properties (2-a) and (2-b). To this end, by combining
     Lemma \ref{lemmecrucialapresdefault}-(b), the fact that $\Delta (V^{p,\mathbb H})=\ ^{p,\mathbb H}(\Delta V)$
     for any process $V$ with locally integrable variation and any filtration $\mathbb H$,
      and $\Delta m=\widetilde Z-Z_{-}$, on ${\Lbrack\tau,+\infty\Rbrack}$ we calculate
\begin{eqnarray*}
&&\Delta L^{\mathbb G}=(1-Z_{-})^{-1}\Delta\widehat{m}+\Delta W^{\mathbb G}-\Delta \left(W^{\mathbb G}\right)^{p,\mathbb G}\\
&=& {{\Delta m}\over{1-Z_{-}}} +{{\Delta\langle m\rangle^{\mathbb F}}\over{(1-Z_{-})^2}}+
{{(\Delta m)^2}\over{(1-\widetilde Z)(1-Z_{-})}}-\ ^{p,\mathbb G}\left({{(\Delta m)^2}\over{(1-\widetilde Z)(1-Z_{-})}}\right)\hskip 0.3cm\\
&=& {{\Delta m}\over{1-\widetilde Z}}+{{\Delta\langle m\rangle^{\mathbb F}}\over{(1-Z_{-})^2}}-
{{^{p,\mathbb F}\left((\Delta m)^2I_{\{\widetilde Z<1\}}\right)}\over{(1-Z_{-})^2}}=
-1+{{1-Z_{-}}\over{1-\widetilde Z}}+\ ^{p,\mathbb F}\left(I_{\{\widetilde Z=1\}}\right)\\
\end{eqnarray*}
Therefore, $1+\Delta L^{\mathbb G}=I_{\Lbrack\tau,+\infty\Rbrack}\Bigl[{{1-Z_{-}}\over{1-\widetilde Z}}+
\ ^{p,\mathbb F}\left(I_{\{\widetilde Z=1\}}\right)\Bigr]+I_{\Rbrack0,\tau\Lbrack}>0.$ This proves the property (2-a).
 In order to prove the property (2-b), we consider a quasi-left-continuous $\mathbb F$-local martingale $M$.
  Then, it is obvious that this quasi-left-continuous assumption implies that $\langle m, M\rangle^{\mathbb F}$
   is continuous and $[X,M]\equiv0$ for any $\mathbb G$-predictable process with finite variation $X$. As a result, we derive
\begin{eqnarray}\label{bracketLGM}
[L^{\mathbb G}, \widehat M]&=&[L^{\mathbb G}, M-M^{\tau}]=(1-Z_{-})^{-1}I_{\Lbrack\tau,+\infty\Rbrack}\is [m, M]+[W^{\mathbb G}, M]\nonumber\\
&=& {1\over{1-Z_{-}}}I_{\Lbrack\tau,+\infty\Rbrack}\is [m, M]
+{{\Delta m}\over{(1-Z_{-})(1-\widetilde Z)}}I_{\Lbrack\tau,+\infty\Rbrack}\is[m, M]\nonumber\\
&=&  \left(1-\widetilde Z\right)^{-1}I_{\Lbrack\tau,+\infty\Rbrack}\is [m, M].
\end{eqnarray}
Therefore, since $[m, M]\in {\cal A}_{loc}(\mathbb F)$, the property (2-b) follows immediately from combining the above equality
 and Lemma \ref{maintechnique}-(b). This ends the proof of the proposition.\qed\end{proof}

\noindent Below, we elaborate our main results deflator for ``the part-after-$\tau$".


\begin{theorem}\label{resultat1apresdefault}
Let $\tau\in {\mathcal H}$ be a finite almost surely  and $L^{\mathbb G}$ be defined by (\ref{LG}). Then, the following assertions hold.\\
(a) If $M$ is a quasi-left-continuous $\mathbb F$-local martingale such that $\sum \Delta M I_{\{\widetilde Z=1>Z_{-}\}}$
 is also an $\mathbb F$-local martingale, then $ {\cal E}(L^{\mathbb G})\left(M-M^{\tau}\right)$ is a $\mathbb G$-local martingale.\\
(b)  For any quasi-left-continuous $\mathbb F$-local martingale, $M$, such that $\{\widetilde Z=1>Z_{-}\}\cap\{\Delta M\not=0\}$
 is evanescent,  $ {\cal E}(L^{\mathbb G})\left(M-M^{\tau}\right)$ is a $\mathbb G$-local martingale.
 \end{theorem}

\begin{proof} Let $M$ be a quasi-left-continuous $\mathbb F$-local martingale such that
$V:=\sum \Delta M I_{\{\widetilde Z=1>Z_{-}\}}$ is an $\mathbb F$-local martingale. As a result, we get
$$
0=\Bigl((1-Z_{-})\is V\Bigr)^{p,\mathbb F}=\left(I_{\{\widetilde Z=1\}}\is [m,M]\right)^{p,\mathbb F}.
$$
Therefore, by combining  this equation, (\ref{bracketLGM}) and Lemma \ref{maintechnique}-(b), we obtain
\begin{eqnarray*}
&&M-M^{\tau}+\langle L^{\mathbb G},M-M^{\tau}\rangle^{\mathbb G}=M-M^{\tau}+
\left({{I_{\Lbrack\tau,+\infty\Rbrack}}\over{1-\widetilde Z}}\is[ m,M]\right)^{p,\mathbb G}\\
&=&M-M^{\tau}+{{I_{\Lbrack\tau,+\infty\Rbrack}}\over{1-Z_{-}}}\is\langle m,M\rangle^{\mathbb F}-
{{I_{\Lbrack\tau,+\infty\Rbrack}}\over{1-Z_{-}}}\is\left(I_{\{\widetilde Z=1\}}\is [ m,M]\right)^{p,\mathbb F}\\
&=&M-M^{\tau}+(1-Z_{-})^{-1}I_{\Lbrack\tau,+\infty\Rbrack}\is\langle m,M\rangle^{\mathbb F}=\widehat M\in {\cal M}_{loc}(\mathbb G).
\end{eqnarray*}
Thus, assertion (a) follows immediately form this combined with Ito's formula applied to $(M-M^{\tau}){\cal E}(L^{\mathbb G})$.
Assertion (b) follows obviously from assertion (a), and the proof of the theorem is completed.
\qed\end{proof}

\noindent As a consequence of this theorem, we describe a class of
$\mathbb F$-quasi-left-continuous processes for which  the NUPBR
property is preserved for the ``part-after-$\tau$".

\begin{corollary}\label{Cor4explDefltor}Suppose that $\tau\in {\mathcal H}$ is finite almost surely,
and $S$ is $\mathbb F$-quasi-left-continuous.
If $\left(S, \sum \Delta S I_{\{\Delta S\not=0\}}\right)$ satisfies the NUPBR$(\mathbb F)$,
 then $S-S^{\tau}$ satisfies the NUPBR$(\mathbb G)$.\end{corollary}

\begin{proof} The proof follows immediately from a combination of Theorem \ref{resultat1apresdefault}-(a),
 Proposition \ref{NUPBRLocalization} (see the appendix),  and the fact that
 \begin{equation}\label{ZtildePandQ} \{\widetilde
Z=1>Z_{-}\}=\{\widetilde Z^Q=1>Z^Q_{-}\}\ \ \ \ \mbox{for any}\
Q\sim P,\end{equation} where $\widetilde Z^Q_t:=Q(\tau\geq
t\big|{\cal F}_t)$ and $Z^Q_t:=Q(\tau> t|{\cal F}_t)$. This last
fact is an immediate application of Theorem 86 of \cite{dm2} by
taking  on the one hand $X=I_{\{\widetilde Z=0\}}$ and $
Y=I_{\{ \widetilde Z^Q=0\}}$, and on the other hand
$X=I_{\{Z_{-}=0\}}$ and $ Y=I_{\{ Z^Q_{-}=0\}}$.\qed
\end{proof}

\section{Proof of Theorems \ref{individualSaftertau0} and \ref{BKPredictableJumpsaftertau3}}
\label{Sectionproofs}


This section focuses on the proofs of Theorems \ref{individualSaftertau0} and \ref{BKPredictableJumpsaftertau3}.
Both proofs are based essentially on the predictable characteristics of $S$ under $\mathbb F$ and $\mathbb G$.
This section is divided into three subsections. The first subsection recalls the predictable characteristics,
and proposes afterwards a functional $\psi$, which is intimately
  related to the set $\{\widetilde Z=1>Z_{-}\}$. This $\psi$ quantifies the part responsible for $\mathbb G$-arbitrages.
   The second and third subsections
   are devoted to the proof of Theorems \ref{individualSaftertau0}
    and \ref{BKPredictableJumpsaftertau3} respectively.

\subsection{Predictable Characteristics of $S$ and the Functional $\psi$}
\noindent To the process $S$, we associate its random measure of
jumps $\mu(dt,dx):=\sum _{u>0} I_{\{\Delta
S_u\not=0\}}\delta_{(u,\Delta S_u)}(dt,dx)$. For any nonnegative
product-measurable functional $H(t,\omega, x)$, we define the
process $H\star\mu$ and   a $\sigma$-finite measure $M_{\mu}^P$ on
the measurable space $\left(\Omega\times{\mathbb
R}^+\times{\mathbb R}^d, {\cal F}_{\infty}\otimes{\cal B}({\mathbb
R }^+)\otimes {\cal B}({\mathbb R}^d)\right)$  by
\begin{equation}\label{integralmuandMpmu}
H\star\mu_t:=\int_0^t\int_{{\mathbb R}^d} H(u,x)\mu(du,dx)\ \ \ \
\mbox{and} \ \ M_{\mu}^P(H):= E\left[H\star\mu_\infty \right].\end{equation} Throughout the rest of the paper,
for any filtration $\mathbb H$, we denote
$$\widetilde{\cal O}(\mathbb H):={\cal O}(\mathbb H)\otimes {\cal
B}({\mathbb R}^d),\ \ \ \  \widetilde{\cal P}(\mathbb H):= {\cal P}(\mathbb H)\otimes {\cal
B}({\mathbb R}^d),$$ and
 $M^P_{\mu}(W|\widetilde{\cal P}(\mathbb H))$, for a nonnegative or bounded functional $W$,
  is the unique $\widetilde{\cal P}(\mathbb H)$-measurable functional $Y$ satisfying
 $M^P_{\mu}(YU)=M^P_{\mu}(WU)$ for any bounded and $\widetilde{\cal P}(\mathbb H)$-measurable functional $U$.
 The random measure $\nu(dt,dx)$ is the unique $\mathbb F$-predictable random measure satisfying
  $\left(H\star\mu\right)^{p,\mathbb F}=H\star\nu$, for any $\widetilde{\cal P}(\mathbb F)$-measurable and nonnegative $H$.
   There is a version of $\nu$
  taking the form of $\nu(dt,dx)=F_t(dx)dA_t$ where $A$ is a nondecreasing and $\mathbb F$-predictable process and $F(dx)$
  is an $\mathbb F$-predictable kernel. Then, the canonical decomposition of $S$ is
  \begin{equation}\label{canonicalS}
  S=S_0+S^c+h\star(\mu-\nu)+b\is A+(x-h)\star\mu,
  \end{equation}
  where $h(x):=xI_{\{\vert x\vert\leq 1\}}$, $S^c$ is the continuous $\mathbb F$-local martingale part of $S$, $b$
   is an $\mathbb F$-predictable process, $h\star(\mu-\nu)$ is the unique pure jumps $\mathbb F$-local martingale
    with jumps given by $h(\Delta S)I_{\{\Delta S\not=0\}}$, and there exists an $\mathbb F$-predictable matrix process, $c$,
     such that $\langle S^c,S^c\rangle^{\mathbb F}= c\is A$.
     Then,
     \begin{equation}\label{PredCharact4SF}
     \mbox{the quadruplet}\ \left(b,c,F,A\right)\  \mbox{is the}\ {\mathbb F}-\mbox{predictable characteristics of}\ S.
     \end{equation}
     These characteristics parameterize the model $(S,\mathbb F, P)$, and will be used frequently throughout the remaining part of the paper.\\

\noindent The following identifies explicitly the source of $\mathbb G$-arbitrage for $S-S^{\tau}$ denoted by  the functional $\psi $,
 and gives some of its properties.
\begin{lemma}\label{lemFunctionalpsi} Consider
\begin{equation}\label{psifunctional}
f_m:=M^P_{\mu}\Bigl(\Delta m|\widetilde{\cal P}(\mathbb F)\Bigr),\ \ \mbox{and}\ \
 \psi:=M^P_{\mu}\left(I_{\{\widetilde Z<1\}}\big|\ {\cal P}(\mathbb F)\right).\end{equation}
Then, the following hold.\\
(a) The process $(f_m)^2\star\mu$ belongs to ${\cal A}_{loc}^+(\mathbb F)$, and there exist
$\beta_m\in L(S^c)$ and $m^{\perp}\in{\cal M}_{loc}(\mathbb F)$ (i.e. $(\beta_m)^{tr}c\beta_m\is A\in{\cal A}^+_{loc}(\mathbb F)$)
such that $[S^c,m^{\perp}]\equiv 0$ and
\begin{equation}\label{GKW4m}
m=m_0+\beta_m\is S^c+m^{\perp}.\end{equation}
(b) We have $\{\psi=0\}=\{Z_{-}+f_m=1\}\subset\{\widetilde Z=1\}$,
$M^P_{\mu}-a.e.$ or equivalently
\begin{equation}\label{iunclusion1}
\{\psi=0\}=\{Z_{-}+f_m=1\}\subset\{\widetilde
Z=1\}\ \ \ \ \mbox{on}\ \{\Delta S\not=0\}.\end{equation}
(c) The nondecreasing process $I_{\{ \psi=0\ \&\ Z_{-}<1\}}\star\mu$ is c\`adl\`ag and $\mathbb F$-locally integrable under any probability measure $Q$.
\end{lemma}
\begin{proof} Assertion (a) is proved in \cite{aksamit/choulli/deng/jeanblanc}. Thus, we address assertions (b) and (c).\\
1) Here, we prove assertion (b). Recall that we always have $E\left[W\star\mu_\infty\right]=
E\Bigl[M^P_{\mu}(W|\widetilde{\cal P}(\mathbb F))\star\nu_\infty\Bigr],$
for any non-negative $\widetilde {\cal O}(\mathbb F)$-measurable functional $W$.
Thus, since $M^P_{\mu}(\widetilde Z|\widetilde{\cal P}(\mathbb F))=Z_{-}+f_m$ and $1-\widetilde Z\leq I_{\{\widetilde Z<1\}}$, we derive
\begin{eqnarray*}
&&0\leq 1-Z_{-}-f_m\leq \psi\ \ \ \ M^P_{\mu}-a.e.\ \ \mbox{and}\\
&&E\Bigl[(1-\widetilde Z)I_{\{ Z_{-}+f_m=1\}}\star\mu_\infty\Bigr]= E\Bigl[(1- Z_{-}-f_m)I_{\{ Z_{-}+f_m=1\}}\star\mu_\infty\Bigr]=0.
\end{eqnarray*}
These clearly prove that, on one hand, we have   $$\{\psi=0\}\subset \{Z_{-}+f_m=1\}\subset\{\widetilde Z=1\},\ \ \ \ M^P_{\mu}-a.e.$$
On the other hand, we derive
$$
E\Bigl[I_{\{Z_{-}+f_m=1\}}\psi\star\mu_{\infty}\Bigr]=E\left[I_{\{Z_{-}+f_m=1\}}I_{\{\widetilde Z<1\}}\star\mu_{\infty}\right]=0.$$
This proves  that $\{Z_{-}+f_m=1\}\subset\{\psi=0\}$, $M^P_{\mu}-a.e.$, and the proof of the assertion (a) is completed.\\
\noindent 2) The proof of the assertion b) follows immediately from Lemma \ref{VFprocess} (where $V^{\mathbb F}$ is defined and
 we recall here $V^{\mathbb F}:= \sum I_{\{\widetilde Z=1>Z_{-}\}}$ for the reader's convenience) and
 the following inequality (which is due to (\ref{iunclusion1}))
 $$
 HI_{\{\psi=0\ \&\ Z_{-}<1\}}\star\mu\leq HI_{\{\widetilde Z=1> Z_{-}\}}\star\mu\leq \sum HI_{\{\widetilde Z=1> Z_{-}\}}=:H\is V^{\mathbb F},$$
 for $H$ nonnegative and bounded. This ends the proof of the lemma.\qed\end{proof}

\noindent We end this subsection with providing the $\mathbb G$-predictable characteristics of $S-S^{\tau}$ as follows.
 Throughout the rest of the paper,
 we put $\mu^{\mathbb G}(dt,dx):=I_{\{t>\tau\}}\mu(dt,dx)$,
 and deduce that $\nu^{\mathbb G}$ --its $\mathbb G$-compensator--
 is given by
\begin{eqnarray}\label{muGnuG}
\nu^{\mathbb G}(dt,dx)&:=&I_{\{t>\tau\}}\Bigl[1-f_m(x,t)(1-Z_{t-})^{-1}\Bigr]\nu(dt,dx).\end{eqnarray}
Furthermore, the $\mathbb G$-canonical decomposition of $S-S^{\tau}$ is given by
\begin{eqnarray*}
S-S^{\tau}&=&\widehat{S^c}+h\star(\mu^{\mathbb G}-\nu^{\mathbb G})+bI_{\Lbrack\tau,+\infty\Rbrack}\is A-
{{c\beta_m}\over{1-Z_{-}}}I_{\Lbrack\tau,+\infty\Rbrack}\is A\\&& \quad -h{{f_m}\over{1-Z_{-}}} I_{\Lbrack\tau,+\infty\Rbrack}\star\nu+
(x-h)\star\mu^{\mathbb G},\end{eqnarray*}
where $\widehat{S^c}$ is defined by (\ref{Mhat}). This decomposition clearly states that
the $\mathbb G$-predictable characteristics of $S-S^{\tau}$, $\left(b^{\mathbb G}, c^{\mathbb G}, F^{\mathbb G},A^{\mathbb G}\right)$,
are given by
\begin{eqnarray}\label{GpredictableCharac}
&&b^{\mathbb G}:=b-\left[\int h(x)f_m(x)F(x)+c\beta_m\right](1-Z_{-})^{-1}I_{\{Z_{-}<1\}},\ \ c^{\mathbb G}:=c\hskip 1cm \nonumber\\
&&F^{\mathbb G}(dx):=\left(1-{{f_m(x,t)}\over{1-Z_{-}}}\right)I_{\{Z_{-}<1\}}F(dx),\ \ A^{\mathbb G}:=A-A^{\tau}.
\end{eqnarray}

\subsection{Proof of Theorem \ref{individualSaftertau0}}
This section proves this theorem. To this end, we start by singling out, in the following remark,  the simplest parts of the theorem, and
the key ideas for the proof of the difficult part(s) as well.

\begin{remark}\label{simplification}
1) Since $(1-Z_{-})^lI_{\{Z_{-}<1\}}$ (for any $l\in\mathbb R$) is $\mathbb F$-locally bounded (see Lemma \ref{localboundZ-}),
 it is easy to see that $I_{\{Z_{-}<1\}}\is X$ satisfies
 the NUPBR$(\mathbb F)$ if and only if $(1-Z_{-})\is X$ does, for any $\mathbb F$-semimartingale $X$. As a result, on one hand,
 the proof of (b)$\Longleftrightarrow$ (c) follows immediately from this fact. On the other hand, as it is mentioned in
 Remark \ref{argum4sectiondeflator},
 it is very clear that
$$S-S^{\tau}={\cal T}_a(S)-\left({\cal T}_a(S)\right)^{\tau}\ \ \ \mbox{and}\ \ \{\Delta {\cal T}_a(S)\not=0\}\cap\{\widetilde Z=1>Z_{-}\}
=\emptyset,$$
and the proof of (b)$\Longrightarrow$ (a) follows from combining these with Corollary \ref{Cor4explDefltor}.
 Therefore, the rest of this
section prepares and delivers afterwards the proof of (a) $\Longrightarrow$ (b), which is the  most technical and difficult part of the theorem.\\
2) The proof of (a) $\Longrightarrow$ (b) relies on applying Theorem \ref{mgDensityCharac} adequately. Hence, the first task in proving this part
resides in we guessing/getting the pair
 $(\beta^{(1)},f^{(1)})$ for $({\cal T}_a(S),\mathbb F)$ from $(\beta^{\mathbb G}, f^{\mathbb G})$ associated to $(S-S^{\tau},\mathbb G)$.
Thanks to Proposition \ref{lemma:predsetFG}, the next lemma prepares the ground for this goal by providing
  equivalent statement to the NUPBR of $(S-S^{\tau},\mathbb G)$, using the $\mathbb F$-predictable functionals only.
   After this step, we will prove that the chosen pair fulfills
   the conditions (\ref{Assump1})-(\ref{Assump2}-(\ref{Assump3}) that correspond to the model $(I_{\{Z_{-}<1\}}\is{\cal T}_a(S),\mathbb F)$.
     \end{remark}

\begin{lemma}\label{mgGtomgF} Let $\Phi_{\alpha}(f)$ (for $\alpha>0$ ) be defined by
\begin{equation}\label{Phialpha}
\Phi_{\alpha}(f):=(f-1)^2I_{\{\vert f-1\vert\leq\alpha\}}+\vert f-1\vert I_{\{\vert f-1\vert>\alpha\}},
\ \ \mbox{for any}\ f\in\widetilde{\cal P}(\mathbb F).
\end{equation}
Then, $(S-S^{\tau})$ satisfies the  NUPBR$(\mathbb G)$ if and only if there exists a pair, $\left(\beta^{\mathbb F},f^{\mathbb F}\right)$, of
$\mathbb F$-predictable process and $\widetilde{\cal P}(\mathbb F)$-predictable functional, such that $f^{\mathbb F}>0\ \ \ M^P_{\mu}-a.e.$,
\begin{equation}\label{conditionsOn(beta,f)1}
(\beta^{\mathbb F})^{tr} c\beta^{\mathbb F}I_{\{Z_{-}<1\}}\is A+
 \Phi_{\alpha}(f^{\mathbb F})(1-Z_{-}-f_m)I_{\{Z_{-}<1\}}\star\mu\in {\cal A}^+_{loc}(\mathbb F),
\end{equation}
and $P\otimes A-a.e.$ on $\{ Z_{-}<1\}$, we have
\begin{eqnarray}
&& \int \vert x f^{\mathbb F}(x)\left(1-{{f_m(x)}\over{1-Z_{-}}}\right)-h(x)\vert F(dx)<+\infty\,
 \ \ \ \ \mbox{and}\ \ \label{conditionsOn(beta,f)2}\\
 &&b+c\left(\beta^{\mathbb F}-{{\beta_m}\over{1-Z_{-}}}\right)+
 \int \left[x f^{\mathbb F}(x)(1-{{f_m(x)}\over{1-Z_{-}}})-h(x)\right]F(dx)\equiv 0.\hskip 0.8cm
\label{conditionsOn(beta,f)3}\end{eqnarray}
\end{lemma}
\begin{proof} In virtue of Theorem \ref{mgDensityCharac}, by using the $\mathbb G$-predictable characteristics of $S-S^{\tau}$,
given in (\ref{GpredictableCharac}), we deduce that $(S-S^{\tau})$ satisfies the NUPBR$(\mathbb G)$ iff
there exists a pair of $\mathbb G$-predictable functionals
$(\beta^{\mathbb G}, f^{\mathbb G})$ such that $f^{\mathbb G}>0$,
\begin{eqnarray}\label{finiteG1}
&&(\beta^{\mathbb G})^{tr}c\beta^{\mathbb G}I_{\Lbrack\tau,+\infty\Rbrack}\is A+
\sqrt{(f^{\mathbb G}-1)^2I_{\Lbrack\tau,+\infty\Rbrack}\star\mu}\in{\cal A}^+_{loc}(\mathbb G),\end{eqnarray}
and $P\otimes A$-a.e. on $\Lbrack\tau,+\infty\Rbrack$,
\begin{eqnarray}
&&\int\vert x f^{\mathbb G}(x)-h(x)\vert(1-Z_{-}-f_m)F(dx)<+\infty,\ \ \ \ \mbox{and}\label{finiteG2}\\
&&0\equiv b+c\left(\beta^{\mathbb G}-{{\beta_m}\over{1-Z_{-}}}\right)+\int
\left[x f^{\mathbb G}(x)(1-{{f_m(x)}\over{1-Z_{-}}})-h(x)\right]F(dx).\hskip 1cm \label{mg4G1}
\end{eqnarray}
Furthermore, to this pair $(\beta^{\mathbb G}, f^{\mathbb G})$ , Proposition \ref{lemma:predsetFG} guarantees the existence of a pair
of $\mathbb F$-predictable functionals $(\beta^{\mathbb F},f^{\mathbb F})$ such that  $f^{\mathbb F}>0$ and
\begin{equation}\label{G=F}
 (\beta^{\mathbb G},f^{\mathbb G})=(\beta^{\mathbb F},f^{\mathbb F})\ \ \ \mbox{on}\ \ \Lbrack\tau,+\infty\Rbrack.
 \end{equation}
Therefore, by inserting (\ref{G=F}) in the three conditions, (\ref{finiteG1}), (\ref{finiteG2}) and (\ref{mg4G1}), and using
afterwards Lemma \ref{maintechnique}--(d) and Proposition \ref{Gstoppingtimeaftertau} (precisely assertions (b), (c) and (d)), the proof
of the lemma follows immediately. \qed  \end{proof}

\begin{remark}\label{remark2difficulty} To guess the pair $(\beta^{(0)}, f^{(0)})$ for the model
 $\left(S^{(0)}:=I_{\{Z_{-}<1\}}\is {\cal T}_a(S),\mathbb F\right)$ from the pair, $(\beta^{\mathbb F},f^{\mathbb F})$,
  provided by the above lemma, we need to derive the $\mathbb F$-predictable characteristics of the model. Thus, we start by getting the
  random measure associated to the jumps of $S^{(0)}$ as $\mu^{(0)}(dt,dx):=I_{\{\widetilde Z<1\ \&\ Z_{-}<1\}}\mu(dt,dx)$, and its $\mathbb F$-
  compensator $\nu^{(0)}$ given by
  \begin{equation}\label{nu0}
  \nu^{(0)}(dt,dx):=\psi(t,x)I_{\{Z_{t-}<1\}}\nu(dt,dx).\end{equation}
Then, by combining this with (\ref{canonicalS}), $hI_{\{\widetilde Z=1>Z_{-}\}}\star\mu\in{\cal A}_{loc}(\mathbb F)$, and
$\left(hI_{\{\widetilde Z=1>Z_{-}\}}\star\mu\right)^{p,\mathbb F}=h\psi I_{\{Z_{-}<1\}}\star\nu$ to derive
the canonical decomposition of $\left(S^{(0)},\mathbb F\right)$ and get its $\mathbb F$-predictable characteristics,
 $\left(b^{(0)},c^{(0)},F^{(0)}(dx), A^{(0)}\right)$, as follows
 \begin{eqnarray}\label{predCharac4S0}
 &&b^{(0)}:=b-\int h(x)\psi(x)F(dx),\ \ \ \  \   c^{(0)}:=c\nonumber\\
 &&F^{(0)}(dx):= \psi(x) F(dx),\ \ \ \ \      A^{(0)}:=I_{\{Z_{-}<1\}}\is A.\hskip1cm
 \end{eqnarray}
Then, in virtue of Theorem \ref{mgDensityCharac}, $\left(S^{(0)},\mathbb F\right)$  satisfies the NUPBR if and only if there exists a pair
$(\beta^{(0)},f^{(0)})$ satisfying $f^{(0)}>0$, (\ref{Assump1}) and (\ref{Assump2}) hold, and after simplifications using (\ref{predCharac4S0}),
\begin{equation}\label{mg4S0}
b+c\beta^{(0)}+\int \left[xf^{(0)}(x)\psi(x)-h(x)\right]F(dx)\equiv 0.
\end{equation}
Therefore, by comparing this equation to (\ref{conditionsOn(beta,f)3}), one can easily conclude that the only pair, $(\beta^{(0)},f^{(0)})$,
 that comes from the pair
$(\beta^{\mathbb F},f^{\mathbb F})$ is given by
$$
\beta^{(0)}:=\left(\beta^{\mathbb F}-{{\beta_m}\over{1-Z_{-}}}\right)I_{\{Z_{-}<1\}},\  \&\
f^{(0)}:=f^{\mathbb G}(x)(1-{{f_m(x)}\over{1-Z_{-}}})\psi^{-1}I_{\{\psi>0\ \&\ Z_{-}<1\}}.
$$
This (apparently) unique choice leads to a major obstacle, as we have no information regarding the integrability of $\psi^{-1} I_{\{\psi>0\}}$.
 However, this also explains that finding an ``equivalent" model that will allow us to control this integrability problem imposes itself. This
 is the aim of the following. \end{remark}

\begin{proposition}\label{NUPBR4T(S)/S1}
Let ${\cal T}_a(S)$ be defined in (\ref{assertion1}) and  consider
\begin{eqnarray}\label{m1S1}
m^{(1)}:=I_{\{\psi=0\ \&\ Z_{-}<1\}}\star(\mu-\nu)\ \ \ \mbox{and}\ \  S^{(1)}:=I_{\{Z_{-}<1\}}\is S-[S,m^{(1)}].\hskip 1cm
\end{eqnarray}
Then, $S^{(0)}:=I_{\{Z_{-}<1\}}\is {\cal T}_a(S)$ satisfies the NUPBR$(\mathbb F)$ if and only if $ S^{(1)}$
satisfies the NUPBR$(\mathbb F)$.
\end{proposition}

\noindent The proof of this proposition is delegated to the appendix for the reader's convenience.
Now, we are in the stage of proving Theorem \ref{individualSaftertau0}.

\begin{proof}{\it of Theorem \ref{individualSaftertau0}} Suppose that $S-S^{\tau}$ satisfies   NUPBR$(\mathbb G)$. Then, due to
Lemma \ref{mgGtomgF}, we deduce the existence of $(\beta^{\mathbb F}, f^{\mathbb F})$ satisfying $f^{\mathbb F}>0$,
 (\ref{conditionsOn(beta,f)1}), (\ref{conditionsOn(beta,f)2}) and (\ref{conditionsOn(beta,f)3}). Thanks to Proposition \ref{NUPBR4T(S)/S1},
 this proof will be completed as soon as we prove that
$(S^{(1)},\mathbb F)$ satisfies the NUPBR. This is the aim of the rest of the proof. To this end, we put $\Sigma_1:=\{Z_{-}<1\ \&\ \psi>0\}$,
 $\widetilde{\Omega}:= \Omega\times[0,+\infty)$,
\begin{eqnarray}\label{beta1f1}
\beta:=(\beta^{\mathbb F}-{{\beta_m}\over{1-Z_{-}}})I_{\{ Z_{-}<1\}},\
 f:=f^{\mathbb F}(x)\left(1-{{f_m(x)}\over{1-Z_{-}}}\right)I_{\Sigma_1}+I_{\widetilde{\Omega}\setminus\Sigma_1}.\hskip 0.5cm \end{eqnarray}
It is obvious that $f>0$. To apply Theorem \ref{mgDensityCharac} using the above pair $(\beta, f)$, we need to derive the predictable characteristics
$(S^{(1)},\mathbb F)$. Thus, we start by getting the random measure for the jumps of this model by
$\mu^{(1)}(dt,dx):=\mu_{S^{(1)}}(dt,dx)=I_{\{\psi(t,x)>0\}}I_{\{Z_{t-}<1\}}\mu(dx,dt)$,
 and its $\mathbb F$-compensator
 \begin{equation}\label{nu1}
 \nu^{(1)}(dt,dx):=I_{\Sigma_1}(tx)\nu(dt,dx),\ \ \ \Sigma_1:=\{\psi>0\ \&\ Z_{-}<1\}.\end{equation}
 Then, again, combining this with (\ref{canonicalS}), and
 $\left(hI_{\{\psi=0\ \&\ Z_{-}<1\}}\star\mu\right)^{p,\mathbb F}=hI_{\{\psi=0\ \&\ Z_{-}<1\}}\star\nu$, we derive easily the canonical
 decomposition of the model and get its predictable characteristics, $\left(b^{(1)},c^{(1)},F^{(1)}(dx),A^{(1)}\right)$,  as follows:
 \begin{eqnarray}\label{PredCharac4S1}
 b^{(1)}:=b-\int h(x)I_{\{\psi(x)=0\}}F(dx),\ \ \ \ c^{(1)}:=c,\nonumber\\
   F^{(1)}(dx):=I_{\{\psi(t,x)>0\}}F(dx),\ \ \ \ A^{(1)}:=I_{\{Z_{-}<1\}}\is A\Bigr).\hskip 1cm \end{eqnarray}
 It is obvious that, by plugging (\ref{beta1f1}) and (\ref{PredCharac4S1}) into (\ref{conditionsOn(beta,f)2}) and (\ref{conditionsOn(beta,f)3}),
   we get
 \begin{eqnarray*}
&&\int \vert xf^{(1)}(x)-h(x)\vert F^{(1)}(dx)<+\infty\ \ \ \ P\otimes A^{(0)}-a.e.\ \ \ \mbox{and}\\
&& b^{(1)}+c\beta^{(0)}+\int \left[xf^{(1)}(x)-h(x)\right]F^{(1)}(dx)\equiv 0,\ \ \ P\otimes A^{(0)}-a.e..\end{eqnarray*}
 Thus, we focus in the rest of this proof on proving the integrability condition (\ref{Assump1}) for the pair ($\beta, f)$. Due to the local
  boundedness of $(1-Z_{-})^{-2}I_{\{Z_{-}<1\}}$ (see Lemma \ref{localboundZ-}) and $(\beta_m)^{tr}c\beta_m\is A+
 (\beta^{\mathbb F})^{tr}c\beta^{\mathbb F}\is A\in{\cal A}^+_{loc}(\mathbb F)$
 (see Lemma \ref{lemFunctionalpsi} and (\ref{conditionsOn(beta,f)1})), we deduce that $\beta^{tr}c\beta\is A\in{\cal A}_{loc}^+(\mathbb F)$.
  Therefore, now we deal with $\sqrt{(f-1)^2\star\mu^{(1)}}\in{\cal A}^+_{loc}(\mathbb F)$. Then,
$$
f-1=(f^{\mathbb F}-1)\left(1-{{f_m(x)}\over{1-Z_{-}}}\right)I_{\Sigma_1}-{{f_m(x)}\over{1-Z_{-}}}I_{\Sigma_1},
\ \ \ \ \Sigma_1:=\{\psi>0\ \&\  Z_{-}<1\}.$$
As a result, since $0\leq 1-Z_{-}-f_m\leq 1$, we obtain
$$\sqrt{(f-1)^2\star\mu}\leq \sqrt{(f^{\mathbb F}-1)^2{{1-Z_{-}-f_m}\over{(1-Z_{-})^2}}I_{\{Z_{-}<1\}}\star\mu}+
\sqrt{{{f_m^2}\over{(1-Z_{-})^2}} I_{\{Z_{-}<1\}}\star\mu}.$$
 Thus, a combination of this with the local boundedness of $(1-Z_{-})^{-2}I_{\{Z_{-}<1\}}$ (see Lemma \ref{localboundZ-}),
 (\ref{conditionsOn(beta,f)1}), and $f_m^2\star\mu\in{\cal A}^+_{loc}(\mathbb F)$
 (see Lemma \ref{lemFunctionalpsi})), the proof of
 $\sqrt{(f-1)^2\star\mu^{(1)}}=\sqrt{(f-1)^2I_{\{\psi>0\ \&\ Z_{-}<1\}}\star\mu}\in{\cal A}^+_{loc}(\mathbb F)$ is completed.
  This ends the proof of the theorem.\qed\end{proof}
\subsection{Proof of Theorem \ref{BKPredictableJumpsaftertau3}}
In virtue of (\ref{muGnuG}) and $\Lbrack\tau,+\infty\Rbrack\subset\{Z_{-}<1\}$, the assumption (\ref{LevyAssumption}) holds iff
\begin{eqnarray*}
0&=&E\Bigl[(1-Z_{-})^{-1}I_{\{Z_{-}+f_m=1>Z_{-}\}}I_{\Lbrack\tau,+\infty\Rbrack}\star\nu_{\infty}\Bigr]\\
&=&E\Bigl[I_{\{Z_{-}+f_m=1>Z_{-}\}}\star\nu_{\infty}\Bigr]=E\Bigl[I_{\{Z_{-}+f_m=1>Z_{-}\}}\star\mu_{\infty}\Bigr].\end{eqnarray*}
This implies that $I_{\{Z_{-}+f_m=1>Z_{-}\}}\star\nu$ and $I_{\{Z_{-}+f_m=1>Z_{-}\}}\star\mu$ are null.
 Thus, we deduce that $m^{(1)}=I_{\{Z_{-}+f_m=1>Z_{-}\}}\star\mu-f_m I_{\{Z_{-}+f_m=1>Z_{-}\}}\star\nu$ is also null.
  Then, the proof of the theorem follows immediately from combining this with Proposition \ref{NUPBR4T(S)/S1}
   and Corollary \ref{corollaryofmain1}--(ii).\qed

\vspace*{0,5cm}

\centerline{\textbf{APPENDIX}}
\appendix
 \normalsize

 \section{Deflators via Predictable Characteristics}
 Most results of this section are elaborated in \cite{aksamit/choulli/deng/jeanblanc},
 and we refer the reader to the appendix of that paper for details. Herein, we consider given a probability $Q$, a filtration $\mathbb H$,
  and a $(\mathbb H, Q)$-quasi-left-continuous semimartingale $X$.
  To this process, we associate the random measure of its jumps, denoted by $\mu_X$, and its $(\mathbb H,Q)$-compensator is denoted by $\nu_X$.
  We suppose that $X$ has the following conical decomposition
  $$X=X_0+X^c+h\star(\mu_X-\nu_X)+(x-h)\star\mu_X+b\is A.$$
  Here $h(x):=xI_{\{\vert x\vert\leq 1\}}$ and $h\star(\mu_X-\nu_X)$ represents the unique pure jumps
  $(\mathbb H,Q)$-local martingale with jumps taking the form of $h(\Delta S)I_{\{\Delta S\not=0\}}$.
   We suppose that $\nu_X(dt,dx)=F(t,dx)dA_t$, and
  $c$ the matrix such that $\langle X^c\rangle=c\is A$.
  The quadruplet $(b, c, F, A)$ is the predictable characteristics of $X$ under $(\mathbb H,Q)$.
   Here the elements of this quadruplet depends on $\left(X,Q,\mathbb H\right)$, but there is no risk of confusion in this part.

\begin{theorem}\label{mgDensityCharac} Let $\left(X,Q,\mathbb H\right)$ be a quasi-left-continuous model,
and $(b^Q, c, F^Q, A)$ be its predictable characteristics under $(\mathbb H, Q)$. Then, $X$ satisfies the NUPBR$(\mathbb H, Q)$
 if and only if there exists a pair $(\beta, f)$, of $\mathbb H$-predictable process $\beta$ and
  $\widetilde{\cal P}(\mathbb H)$-measurable functional  $f$, such that
\begin{eqnarray}
&&f>0,\ \ \beta^{tr}c\beta\is A+ \sqrt{(f-1)^2\star\mu_X}\in {\cal A}^+_{loc}(\mathbb H,Q),\label{Assump1}\\
&&\int \vert xf(x)-h(x)\vert F(dx)<+\infty,\ \ \ Q\otimes A-a.e.\label{Assump2}\\
 &&b+c\beta+\int \left[xf(x)-h(x)\right]F(dx)=0,\ \ \ Q\otimes A-a.e.\label{Assump3}
 \end{eqnarray}
 \end{theorem}
\noindent See \cite{aksamit/choulli/deng/jeanblanc} for the proof.

\begin{proposition}\label{NUPBRLocalization}
Let $X$ be an $\mathbb H$-adapted process. Then, the following assertions are equivalent.\\
{\rm{(a)}}  There exists a sequence   $(T_n)_{n\geq 1}$ of
$\mathbb H$-stopping times that increases to $+\infty$, such that
for each $n\geq 1$, there exists a probability $Q_n$ on $(\Omega,
{\cal H}_{T_n})$ such that $Q_n\sim P$ and $X^{T_n}$ satisfies
 NUPBR$(\mathbb H)$ under $Q_n$.\\
{\rm{(b)}}  $X$ satisfies   NUPBR$(\mathbb H)$.\\
{\rm{(c)}} There exists an $\mathbb H$-predictable process $\phi$,
such that $0<\phi\leq 1$ and  $(\phi\is X)$ satisfies
NUPBR$(\mathbb H)$.
\end{proposition}

\noindent The proof of this proposition can be found in Aksamit et al. \cite{aksamit/choulli/deng/jeanblanc}.\\

\begin{proposition}\label{lemma:predsetFG} Suppose that $\tau$ is a honest time, and let $H^{\mathbb G}$ be an $\widetilde{\cal P}(\mathbb
G)$-measurable functional. Then, the following assertions
 hold.\\
{\rm{(a)}}  There exist two $\widetilde{\cal
P}(\mathbb F)$-measurable functional $H^{\mathbb F}$ and $K^{\mathbb F}$ such
that
\begin{eqnarray}\label{eq:widePGandwidePF}
H^{\mathbb G}(\omega,t,x) = H^{\mathbb F}(\omega,t,x) I_{\Lbrack
0,\tau\Lbrack}+ K^{\mathbb F}(\omega,t,x)I_{\Lbrack
\tau,+\infty\Rbrack}.
\end{eqnarray}
{\rm{(b)}}  If furthermore $H^{\mathbb G}>0$ (respectively $H^{\mathbb
G}\leq 1$), then we can choose $K^{\mathbb F}>0$ (respectively
$K^{\mathbb F}\leq 1$) in (\ref{eq:widePGandwidePF}).\\
\end{proposition}

\begin{proof} The proofs of assertions (a) and (b) follow from mimicking Jeulin's proof \cite[Proposition
5,3]{Jeu}, and will be omitted herein.
\end{proof}

 \section{$\mathbb G$-local Integrability versus $\mathbb F$-local Integrability}\label{SubsectionLocalisation}
This subsection connects the $\mathbb G$-localisation and the $\mathbb F$-localisation for the part after $\tau$.
This completes the analysis of  \cite{aksamit/choulli/deng/jeanblanc} regarding the issue of local integrability under $\mathbb F$
 and $\mathbb G$, where the part up to $\tau$ is fully discussed.
 There is a major difference between the current results and those of \cite{aksamit/choulli/deng/jeanblanc},
  which lies in the fact that for the case up to $\tau$ we loose information after an $\mathbb F$-stopping when we pass to $\mathbb F$.
   However, for the part after $\tau$, as long as $\tau$ is finite, we pass from $\mathbb G$-localisation to
    $\mathbb F$-localisation without any loss of information. The following is the most innovative result of the appendix.
\begin{proposition}\label{lemma:predsetFG1}  The following assertions hold.\\
{\rm{(a)}} If $\tau$ is a finite almost surely honest time and $(\sigma_n^{\mathbb G})_{n\geq 1}$ is a
 sequence of finite $\mathbb G$-stopping times that increases to infinity, then there exists a sequence
  of finite $\mathbb F$-stopping times, $(\sigma_n^{\mathbb F})_{n\geq 1}$, that increases to infinity as well and
\begin{equation}\label{SigmaGSigmaF}
\max(\sigma_n^{\mathbb G},\tau)=\max(\sigma_n^{\mathbb F},\tau),\ \ P-a.s.
\end{equation}
{\rm{(b)}} If $\tau\in {\cal H}$ and is finite almost surely, then there exists a sequence of $\mathbb F$-stopping times,
 $(\sigma_n)_{n\geq 1}$, that increases to infinity almost surely and
\begin{equation}\label{bounded4Z-}
\Bigl\{Z_{-}<1\Bigr\}\cap\Lbrack 0,\sigma_n\Lbrack\subset\Bigl\{1-Z_{-}\geq {1\over{n}}\Bigr\},\ \ \ \forall\ n\geq 1.
\end{equation}
Or equivalently, $(1-Z_{-})^{-1}I_{\{Z_{-}<1\}}$ is $\mathbb F$-locally bounded when $\tau\in {\cal H}$ and is finite almost surely.
\end{proposition}

\begin{proof} The proof of this proposition is given in two parts where we prove assertions (a) and (b) respectively.\\
1) The proof of assertion (a)  boils down to the following fact:
 \begin{eqnarray}
\mbox{for any $\mathbb G$-stopping time,}&&\sigma^{\mathbb G},\ \mbox{there exists an $\mathbb F$-stopping time,}
\ \sigma^{\mathbb F}\ \mbox{ such that}\nonumber\\
&&\sigma^{\mathbb G}\vee \tau=\sigma^{\mathbb F}\vee \tau\ \ \ \ P-a.s.\label{mainfact}
\end{eqnarray}
Indeed, if this fact holds, then there exists $\mathbb F$-stopping times, $(\sigma_n)_{n\geq 1}$ such that for any $n\geq 1$,
  the pair $(\sigma_n^{\mathbb G},\sigma_n)$ satisfies (\ref{SigmaGSigmaF}). Since $\sigma_n^{\mathbb G}$ increases with $n$,
   by putting $\sigma_n^{\mathbb F}:=\sup_{1\leq k\leq n}\sigma_k$, we can easily prove that the pair
    $(\sigma_n^{\mathbb G},\sigma_n^{\mathbb F})$ satisfies (\ref{SigmaGSigmaF}) as well (this is due to
    $\max_{1\leq i\leq n}(x_i\vee y)=(\max_{1\leq i\leq n}x_i)\vee y$ for any nonnegative $x_i, y$). Then,
     assertion (a) follows immediately from taking the limit in (\ref{SigmaGSigmaF}) and making use of $\tau<+\infty$ P-a.s.
      which implies that $\sup_{n\geq 1}\sigma_n=\lim_{n\longrightarrow+\infty}\sigma_n^{\mathbb F}=+\infty$ P-a.s.
       This shows that the proof of assertion (a) is achieved as long as we prove the claim (\ref{mainfact}).
        This is the main focus of the remaining part of this proof. \\
By applying the proposition below (which is fully due to Barlow \cite{Barlow}) to the process
 $Y^{\mathbb G}=I_{\Rbrack \sigma^{\mathbb G}\vee \tau,+\infty\Rbrack}$, we obtain the existence of an
  $\mathbb F$-progressively measurable process $K^{\mathbb F}$ such that
\begin{equation}\label{processK}
Y^{\mathbb G}=K^{\mathbb F}I_{\Rbrack\tau,+\infty\Rbrack}.\end{equation}
Thus, it is easy that one can replace $K^{\mathbb F}$ with $I_{\{K^{\mathbb F}=1\}}$.
Since $\Lbrack\tau,+\infty\Rbrack\subset\{Z<1\}$, It is also easy to check that one can choose $K^{\mathbb F}$ such that,
 on $\{\tau<\sigma^{\mathbb G}$, $\{K^{\mathbb F}=1\}\subset\{Z<1\}$. Then, put
\begin{equation}\label{sigma}
\sigma:=\inf\{t\geq 0:\ \ K^{\mathbb F}_t=1\}.
\end{equation}
This is an $\mathbb F$-stopping time, and due to $\Rbrack \sigma^{\mathbb G}\vee \tau,+\infty\Rbrack\subset\{K^{\mathbb F}=1\}$, we get
\begin{equation}\label{sigma1}
\sigma\leq \tau\vee \sigma^{\mathbb G}\ \ \ P-a.s.
\end{equation}
By applying Proposition \ref{Barlow}, we deduce the existence of two double sequence of
$\mathbb F$-stopping times $(\alpha_{nm})_{n,m\geq 1}$ and $(\beta_{nm})_{n,m\geq 1}$ satisfying the four assertions of the proposition. As a
result, we get, on $\{\tau<\sigma^{\mathbb G}$, we have
$$\Rbrack\sigma^{\mathbb G},+\infty\Rbrack\subset\{K^{\mathbb F}=1\}\subset\bigcup_{n,m\geq 1}\Rbrack\alpha_{nm},\beta_{nm}\Rbrack.$$
By combining this with (\ref{sigma}), we deduce that
$$
\{\tau<\sigma^{\mathbb G}\}\subset \bigcup_{n,m\geq 1}\{\alpha_{nm}\leq \sigma\leq \sigma^{\mathbb G}<\beta_{nm}\}.
$$
Thanks to assertions (i) and (ii) of Proposition \ref{Barlow}, that claims that $\tau$ takes values in $[\beta_{n(m-1)},\alpha_{nm}[$ only,
 we get
$\{\alpha_{nm}\leq \sigma\leq \sigma^{\mathbb G}<\beta_{nm}\}=\{\tau<\alpha_{nm}\leq \sigma\leq \sigma^{\mathbb G}<\beta_{nm}\}$, and hence
\begin{equation}\label{inclusion1}
\{\tau<\sigma^{\mathbb G}\}\subset \bigcup_{n,m\geq 1}\{\tau<\alpha_{nm}\leq \sigma\leq \sigma^{\mathbb G}<\beta_{nm}\}.
\end{equation}
Now, due to (\ref{processK}) and the fact that on $[\sigma, \sigma+\epsilon[\cap \{K^{\mathbb F}=1\}\not=\emptyset$ $P-a.s$, we deduce that
$\left(\tau\leq \sigma<\sigma^{\mathbb G}\right)$ is an impossible event. Therefore, (\ref{inclusion1}) leads to 
$$
\{\tau<\sigma^{\mathbb G}\}\subset \{\sigma=\sigma^{\mathbb G}\}.
$$
Hence, by combing this with (\ref{sigma1}), we derive 
\begin{eqnarray*}
\tau\vee\sigma^{\mathbb G}
&=&(\tau\vee\sigma^{\mathbb G})I_{\{\sigma^{\mathbb G}\leq\tau\}}+(\tau\vee\sigma^{\mathbb G})I_{\{\tau<\sigma^{\mathbb G}\}}
= \tau I_{\{\sigma^{\mathbb G}\leq\tau\}}+(\tau\vee\sigma)I_{\{\tau<\sigma^{\mathbb G}\}}\\
&=&(\tau\vee\sigma)I_{\{\sigma^{\mathbb G}\leq\tau\}}+(\tau\vee\sigma)I_{\{\tau<\sigma^{\mathbb G}\}}=\tau\vee\sigma.
\end{eqnarray*}
This proves (\ref{mainfact}), and the prof of assertion (a) is completed.\\
2) Here, we prove assertion (b). Since $\tau\in {\cal H}$, then $(1-Z_{-})^{-1}I_{\Lbrack
\tau,+\infty\Rbrack}$ is $\mathbb G$-locally bounded due to Lemma \ref{localboundZ-}-(b).
 Thus, on the one hand, there exists a sequence of $\mathbb G$-stopping times, $(\sigma_n^{\mathbb G})_{n\geq 1}$
  that increases to infinity almost surely and
\begin{equation}\label{equa2000}
{\Lbrack\tau,+\infty\Rbrack}\cap {\Rbrack 0,\sigma_n^{\mathbb G}\Lbrack}\subset\Bigl\{1-Z_{-}\geq 1/n\Bigr\}.
\end{equation}
On the other hand, thanks to assertion (a), there exists a sequence of $\mathbb F$-stopping times, $(\sigma_n)_{n\geq 1}$
 that increases to infinity almost surely and satisfies (\ref{SigmaGSigmaF}). Then, by inserting this in (\ref{equa2000}), we get
$$
{\Lbrack\tau,+\infty\Rbrack}\cap {\Rbrack 0,\sigma_n\Lbrack}\subset\Bigl\{1-Z_{-}\geq 1/n\Bigr\}.$$
By taking the $\mathbb F$-predictable projection on both side, we get
$$
0\leq (1-Z_{-})I_{\Rbrack 0,\sigma_n\Lbrack}\leq I_{\Bigl\{1-Z_{-}\geq 1/n\Bigr\}}.
$$ This implies that $\Bigl\{1-Z_{-}<1/n\Bigr\}\subset\{Z_{-}=1\}\cup \Lbrack\sigma_n,+\infty\Rbrack$, which is equivalent to 
(\ref{bounded4Z-}). Hence, the proof of assertion (b) is achieved and that of the proposition as well.\qed\end{proof}

\begin{proposition}\label{Barlow}
Suppose that $\tau$ is a honest time. Then, the following hold.\\
(i) There exists two double sequences of $\mathbb F$-stopping times $(\alpha_{n,m})_{n,m\geq 1}$ and $(\beta_{n,m})_{n,m\geq 1}$ such that $\alpha_{n,m}\leq \beta_{n,m}$ P-a.s. for all $n,m\geq 1$, and
\begin{equation}\label{tauincludedinAlphanm}
\Lbrack\tau,+\infty\Rbrack\subset\{Z<1\}\subset\bigcup_{n,m\geq 1} \Rbrack\alpha_{n,m},\beta_{n,m}\Rbrack.
\end{equation}
(ii) For any $n,m\geq 1$, $\{\tau\geq\alpha_{nm}\}\subset\{\tau\geq\beta_{nm}\}$ $P-a.s.$\\
(iii) For any $\mathbb G$-optional process $Y^{\mathbb G}$, there exists an $\mathbb F$-progressively measurable process $K^{\mathbb F}$ such that
 \begin{equation}\label{YGanddKF}
 Y^{\mathbb G}I_{\Rbrack\tau,+\infty\Rbrack}=K^{\mathbb F}I_{\Rbrack\tau,+\infty\Rbrack}.\end{equation}
(iv) For any $\mathbb G$-optional c\`adl\`ag process $Y^{\mathbb G}$ such that $Y^{\mathbb G}=0$ on $\Rbrack0,\alpha_{n,m}\Rbrack$ and constant on $\Rbrack\beta_{n,m},+\infty\Rbrack$, there exists an $\mathbb F$-progressively measurable process $K^{\mathbb F}$ that is c\`adl\`ag and satisfies (\ref{YGanddKF}).
 \end{proposition}

 \begin{proof} For the proof we refer the reader to \cite{Barlow}. In fact, assertion (i) is exactly Lemma 4.1-(iv) in \cite{Barlow},
 while the assertion (ii) is a combination of Proposition 4.3 and Lemma 4.4-(ii) of the same paper.
 \end{proof}

 The next result addresses the $\mathbb G$-local integrability involving the random 
  measures that is vital for the proof of Theorem \ref{individualSaftertau0}.

\begin{proposition}\label{Gstoppingtimeaftertau}
  Suppose that  $\tau\in {\mathcal H}$ is  finite almost surely. Let $\Phi_{\alpha}(.)$ (for $\alpha>0$) be defined in (\ref{Phialpha}).
   Then, the following properties hold.\\
   {\rm{(a)}} Let $f$ be a real-valued and $\widetilde{\cal P}(\mathbb H)$-measurable functional.
    Then,\\  $\sqrt{(f-1)^2\star \mu}$ belongs to $ {\cal A}^+_{loc}(\mathbb H)$ if and only if
 $\Phi_{\alpha}(f)\star \mu \in {\cal
A}^+_{loc}(\mathbb H)$ does.\\
{\rm{(b)}} Let $f$  be a real-valued and $\widetilde{\cal P}(\mathbb H)$-measurable functional.
    Then,\\
 $\sqrt{(f-1)^2I_{\Lbrack \tau,+\infty \Rbrack}\star \mu} \in {\cal A}^+_{loc}(\mathbb G)$ iff
$\Phi_{\alpha}(f)(1-Z_{-}-f_m) I_{\{ Z_{-}<1\}}\star {\mu}\in {\cal A}^+_{loc}(\mathbb F)$.\\
{\rm{(c)}} Let $\phi$ be nonnegative and $\mathbb F$-predictable process. Then, $P\otimes A-a.e.$\\
$\Lbrack\tau,+\infty\Rbrack\subset\{\phi<+\infty\}$ if and only if $\{Z_{-}<1\}\subset\{\phi<+\infty\}$.\\
{\rm{(d)}} Let $\phi$ be an $\mathbb F$-predictable process. Then, $P\otimes A$-a.e. \\
$\Lbrack\tau,+\infty\Rbrack\subset \{\phi=0\}$
 if and only if $\{Z_{-}<1\}\subset \{\phi=0\}$.
\end{proposition}

\begin{proof}
(a) Assertion (a) is borrowed from \cite{aksamit/choulli/deng/jeanblanc} ( see Proposition C.3--(a)).\\
(b) Thanks to assertion (a), we deduce that $\sqrt{(f-1)^2I_{\Lbrack \tau,+\infty \Rbrack}\star \mu} \in {\cal A}^+_{loc}(\mathbb G)$ iff
$\Phi_{\alpha}(f)\star\mu^{\mathbb G}\in {\cal A}^+_{loc}(\mathbb G)$ iff
\begin{equation}\label{equa9999}
\Phi_{\alpha}(f)\left(1-{{f_m}\over{1-Z_{-}}}\right)I_{\Lbrack\tau,+\infty\Rbrack}\star\nu=
\Phi_{\alpha}(f)\star\nu^{\mathbb G}\in {\cal A}^+_{loc}(\mathbb G).\end{equation}
 Then, a direct application of Lemma \ref{maintechnique}--(d) to the pair
 $$(\varphi, V):=\Bigl([1-f_m(1-Z_{-})^{-1}]I_{\{Z_{-}<1\}},\ \Phi_{\alpha}(f)\star\nu\Bigr),$$ the proof of assertion (b) follows immediately. \\
  (c) Suppose that $P\otimes A$-a.e. that $\Lbrack\tau,+\infty\Rbrack\subset\{\phi<+\infty\}$.
   This is equivalent to $I_{\Lbrack\tau,+\infty\Rbrack}\leq I_{\{\phi<+\infty\}}$ $P\otimes A$-a.e..
    Then, by taking the $\mathbb F$-predictable projection on both sides, we get $1-Z_{-}\leq I_{\{\phi<+\infty\}}$ $P\otimes A$-a.e..
    This obviously  proves that $\Lbrack\tau,+\infty\Rbrack\subset\{\phi<+\infty\}$ implies $\{Z_{-}<1\}\subset\{\phi<+\infty\}$.
    The reverse sense follows from $\Lbrack\tau,+\infty\Rbrack\subset\{Z_{-}<1\}$. This ends the proof of assertion (c).\\
(d) The proof of assertion (d) mimics the proof of assertion (c), and will be omitted. This ends the proof of the proposition. \qed\end{proof}

\section{Proofs for Lemmas \ref{lemmecrucialapresdefault} and \ref{maintechnique} of Subsection\ref{Stochastic1}}

\begin{proof} {\it of Lemma \ref{lemmecrucialapresdefault}} The proof of the lemma will be achieved in three steps.

\noindent {\bf 1)} This step proves assertion  (a). From
Lemma \ref{localboundZ-}
$$
I_{\Lbrack\tau,+\infty\Rbrack}\is V  -
I_{\Lbrack\tau,+\infty\Rbrack}\is V ^{p,\mathbb F}  +
I_{\Lbrack\tau,+\infty\Rbrack}(1-Z_{-})^{-1}\centerdot\langle V,
\mm\rangle^{\mathbb F}
$$
is a $\mathbb G$-local martingale, hence
\begin{eqnarray*}
  \left(I_{\Lbrack\tau,+\infty\Rbrack}\is V\right)^{p,\mathbb G} &=&   I_{\Lbrack\tau,+\infty\Rbrack}\is V ^{p,\mathbb F}
    - I_{\Lbrack\tau,+\infty\Rbrack}(1-Z_{-})^{-1}\centerdot\langle V,  \mm\rangle^{\mathbb F} \\
  \\
  &=& I_{\Lbrack\tau,+\infty\Rbrack}\is V ^{p,\mathbb F}  - 
  I_{\Lbrack\tau,+\infty\Rbrack}(1-Z_{-})^{-1}\centerdot\left(\Delta m\is V\right)^{p,\mathbb F} \\
  \\
&=&
I_{\Lbrack\tau,+\infty\Rbrack}(1-Z_{-})^{-1}\centerdot\Bigl((1-Z_{-}-\Delta
m)\is V\Bigr)^{p,\mathbb F}\,,
\end{eqnarray*}
where the second equality follows from Yoeurp's lemma. This ends the proof of (\ref{Gcompensatoraftertau}).
 The equality (\ref{predictableprojectionsjumpsaftertau}) follows immediately from (\ref{Gcompensatoraftertau})
  by taking the jumps in both sides, and using $\Delta \left(K^{p,\mathbb H}\right)=\ ^{p,\mathbb H}(\Delta K)$ when both terms exist.\\

\noindent{\bf 2)} Now, we  prove assertion (b). By applying
(\ref{predictableprojectionsjumpsaftertau}) for $ V_{\epsilon,\delta}\in {\cal
A}_{loc}(\mathbb F)$ given by
 $$
 V_{\epsilon,\delta}:=\sum(\Delta M) (1-\widetilde Z)^{-1}I_{\{\vert\Delta M\vert\geq \epsilon,\  1-\widetilde Z\geq\delta\}},$$
 we get, on $\Lbrack\tau,+\infty\Rbrack$,
 $$
  ^{p,{\mathbb G}}\left((\Delta M) (1-\widetilde Z)^{-1}I_{\{\vert\Delta M\vert\geq \epsilon,\  { 1-\widetilde Z\geq \delta}\}}\right)=
 {(1-Z_{-})^{-1}}\ ^{p,{\mathbb F}}\left({\Delta M}\,I_{\{\vert\Delta M\vert\geq \epsilon,\  1-\widetilde Z\geq\delta\}}\right).$$ 
 Then, the first equality in (\ref{inclusionofsetsaftertau}) follows from letting $\epsilon$ and $\delta$ go to zero, and we get on
  $\Lbrack\tau,+\infty\Rbrack$
  $$ ^{p,{\mathbb G}}\left({{\Delta M}\over{ 1-\widetilde Z}}\right)= 
  {(1-Z_{-})^{-1}}\ ^{p,{\mathbb F}}\left({\Delta M}\,I_{\{ 1-\widetilde Z>0\}}\right)=
  {(1-Z_{-})^{-1}}\ ^{p,{\mathbb F}}\left({\Delta M}\,I_{\{\widetilde Z<1\}}\right).$$
   To prove the second equality in
  (\ref{inclusionofsetsaftertau}), we write that,
  on $\Lbrack\tau,+\infty\Rbrack$,
 \begin{eqnarray*}
    ^{p,{\mathbb G}}\left({1\over{ 1-\widetilde Z}}\right)&=&(1-Z_{-})^{-1}+(1-Z_{-})^{-1}\ ^{p,{\mathbb G}}\left({{\Delta m}\over{1-\widetilde Z}}\right)\\
    \\
  &=&(1-Z_{-})^{-1}+(1-Z_{-})^{-2}\ ^{p,{\mathbb F}}\left({{\Delta m}}I_{\{ 1-\widetilde Z>0\}}\right)\\
    \\
&=&(1-Z_{-})^{-1}-(1-Z_{-})^{-1}\ ^{p,{\mathbb F}}\left(I_{\{ \widetilde Z=1\}}\right )
     =(1-Z_{-})^{-1}\ ^{p,{\mathbb F}}\left(I_{\{\widetilde Z<1\}}\right).\end{eqnarray*}
The second equality is due to
(\ref{predictableprojectionsjumpsaftertau}), and the third equality follows from combining $\,^{p,{\mathbb F}}(\Delta m)=0$,
and $\Delta m=\widetilde Z-Z_{-}$. This ends the proof of assertion (b).\\

   \noindent{\bf 3)} The proof of (\ref{UsefulProjections}) follows immediately from assertion (b)
and the fact that the thin process $\,^{p,{\mathbb
F}}\left({\Delta M}\,I_{\{ \widetilde Z<1\}}\right)$ may take nonzero values on
countably many predictable stopping times only, on which $\Delta
M$ already vanishes. This completes the proof of the lemma.\qed
 \end{proof}

\begin{proof}{\it Lemma \ref{maintechnique}} The proof of the lemma is given in three parts.
 In the first part we prove both assertions (a) and (b), while in the second and the third
  parts we focus on assertions (c) and (d) respectively.\\
 {\bf 1)} Let $V$ be an $\mathbb F$-adapted process with finite variation. Then, we obtain
$$
{\rm{Var}}(U)=\left(1-Z\right)^{-1}I_{\Lbrack
\tau,+\infty\Rbrack}\is {\rm{Var}}(V).$$ Therefore, since
$1-{\widetilde Z}_t=P(\tau<t|{\cal F}_t)\leq 1-Z_t$, for any
bounded and $\mathbb F$-optional process $\phi$ such that $\phi\is
{\rm{Var}}(V)\in {\cal A}^+(\mathbb F)$, we obtain
\begin{eqnarray}\label{inequalityVU}
&&E\Bigl[(\phi\is{\rm{Var}}(U))_{\infty}\Bigr]=\displaystyle E\left(\int_0^{\infty} {{ \phi_t I_{\{t>\tau\}}}\over{1- Z_t}}d {\rm{Var}}(V)_t\right)\nonumber\\
\nonumber\\
 &= &E\left(\displaystyle \int_0^{\infty} {{\phi_t P(\tau<t|{\cal F}_t)}\over{1- Z_t}}I_{\{ Z_t<1\}}d {\rm{Var}}(V)_t\right) \leq
E\Bigl[(\phi\is {\rm{Var}}(V))_{\infty}\Bigr].\end{eqnarray}
As a result, by taking $\phi=I_{\Lbrack 0,\sigma\Rbrack}$ in (\ref{inequalityVU}), for an $\mathbb F$-stopping time $\sigma$ such that Var$(V)^{\sigma-}\in {\cal A}^+(\mathbb F)$, we get $E\Bigl[{\rm{Var}}(U)_{\sigma-}\Bigr]\leq E\Bigl[Var(V)_{\sigma-}\Bigr]$. This proves that the process $U$ has a finite variation and hence is well defined as well. Being $\mathbb G$-adapted for $U$ is obvious, while being c\`adl\`ag follows immediately from (\ref{inequalityVU}). This ends the proof of assertion (a).\\
To prove assertion (b), we assume that $V\in {\cal A}_{loc}(\mathbb F)$ and consider $(\vartheta_n)_{n\geq 1}$, a sequence of $\mathbb F$-stopping times that increases to $+\infty$ such that ${\rm{Var}}(V)^{\vartheta_n}\in {\cal A}^+(\mathbb F)$. Then, by choosing $\phi=I_{\Lbrack 0,\vartheta_n\Lbrack}$ in (\ref{inequalityVU}), we conclude that $U$ belongs to ${\cal A}_{loc}(\mathbb G)$ whenever $V$ does under $\mathbb F$. For the case when $V\in{\cal A}(\mathbb G)$, it is enough to take $\phi=1$ in (\ref{inequalityVU}), and conclude that $U\in {\cal A}(\mathbb G)$. To prove (\ref{GFcompensator1}), for any $n\geq 1$, we put
$$
U_n:=\left(1-Z\right)^{-1}I_{\Lbrack \tau,+\infty\Rbrack}I_{\{ \widetilde Z\leq 1-{1\over{n}}\}}\is V=\left(1-\widetilde Z\right)^{-1}I_{\Lbrack \tau,+\infty\Rbrack}I_{\{ \widetilde Z\leq 1-{1\over{n}}\}}\is V,\ \ \ n\geq 1.$$
Then, thanks to (\ref{Gcompensatoraftertau}), we derive
\begin{eqnarray*}
&& U^{p,\mathbb G}=\displaystyle\lim_{n\longrightarrow+\infty}\left(U_n\right)^{p,\mathbb G}=
\displaystyle\lim_{n\longrightarrow +
\infty}(1-Z_{-})^{-1}I_{\Lbrack \tau,+\infty\Rbrack}\is\Bigl(I_{\{ \widetilde Z\leq 1-{1\over{n}}\}}\is V\Bigr)^{p,\mathbb F}.\end{eqnarray*}

This clearly implies (\ref{GFcompensator1}).\\
{\bf 2)} It is easy to see that it is enough to prove the assertion for the case when $V$ is nondecreasing.
 Thus, suppose that $V$ is nondecreasing. It obvious that $(1-\widetilde Z)\is V\in {\cal A}^+_{loc}(\mathbb F)$
  implies $I_{\Lbrack \tau,+\infty\Rbrack}\is V\in {\cal A}^+_{loc}(\mathbb G)$. Hence, for the rest of this part,
   we focus on proving the reverse. Suppose $I_{\Lbrack \tau,+\infty\Rbrack}\is V\in {\cal A}^+_{loc}(\mathbb G)$.
    Then, there exists a sequence $\mathbb G$-stopping times that increases to infinity and
     $\left(I_{\Lbrack \tau,+\infty\Rbrack}\is V\right)^{\sigma_n^{\mathbb G}}\in {\cal A}^+(\mathbb G).$
      Thanks to Proposition \ref{lemma:predsetFG}-(c), we obtain a sequence of $\mathbb F$-stopping times,
       $(\sigma^{\mathbb F}_n)_{n\geq 1}$, that increases to infinity and $\sigma^{\mathbb G}_n\vee\tau=\tau\vee\sigma_n^{\mathbb F}$.
        Therefore, we get
  $\left(I_{\Lbrack \tau,+\infty\Rbrack}\is V\right)^{\sigma_n^{\mathbb G}}\equiv
  \left(I_{\Lbrack \tau,+\infty\Rbrack}\is V\right)^{\sigma_n^{\mathbb F}}$ and hence
\begin{equation}\label{mainequality00}
E\left((1-\widetilde Z)\is V_{\sigma_n^{\mathbb F}}\right)=E\left(I_{\Lbrack \tau,+\infty\Rbrack}\is V_{\sigma_n^{\mathbb F}}\right)
=E\left(I_{\Lbrack \tau,+\infty\Rbrack}\is V_{\sigma_n^{\mathbb G}}\right)<+\infty.\end{equation}
This proves that the process $(1-\widetilde Z)\is V$ belongs to ${\cal A}^+_{loc}(\mathbb F)$, and the proof of assertion (c) is achieved.\\
3) The proof of assertion (d) follows all the steps of the proof of assertion (c), except (\ref{mainequality00}) which takes the form of
$$
E\left((1-Z_{-})\varphi\is V_{\sigma_n^{\mathbb F}}\right)=E\left(I_{\Lbrack \tau,+\infty\Rbrack}\varphi\is V_{\sigma_n^{\mathbb F}}\right)
=E\left(I_{\Lbrack \tau,+\infty\Rbrack}\varphi\is V_{\sigma_n^{\mathbb G}}\right)<+\infty$$
instead due to the predictability of $V$. This proves that $I_{\Lbrack \tau,+\infty\Rbrack}\varphi\is V\in{\cal A}^+_{loc}(\mathbb G)$ 
if and only if $(1-Z_{-})\varphi\is V\in{\cal A}^+_{loc}(\mathbb F)$, while  the equivalence 
$(1-Z_{-})\varphi\is V\in{\cal A}^+_{loc}(\mathbb F)$ iff  $I_{\{Z_{-}<1\}}\varphi\is V\in{\cal A}^+_{loc}(\mathbb F)$
 follows from the fact that  $(1-Z_{-})^{-1}I_{\{Z_{-}<1\}}$ is $\mathbb F$-locally bounded (see Proposition \ref{lemma:predsetFG1}-(b)). This 
 ends the proof of assertion (d) and the proof of the lemma as well. \qed
\end{proof}

\section{Proof of Proposition \ref{NUPBR4T(S)/S1}} The proof relies essentially on an adequate application(s) of Theorem \ref{mgDensityCharac}.
To this end, we consider
 $$Z^{(\psi)}:={\cal E}(N^{(\psi)})\ \ \mbox{where}\ \ N^{(\psi)}:=(\psi-1)I_{\{\psi>0\}}\star(\mu-\nu),$$
and remark that $Z^{(\psi)}$ is a positive $\mathbb F$-local martingale. Hence, in virtue of Proposition \ref{NUPBRLocalization},
we can assume without loss of generality that $Z^{(\psi)}$ is a uniformly
 integrable martingale,
 and put $Q:=Z_{\infty}^{\psi}\cdot P$ (probability measure equivalent to $P$). Thus, the rest of the proof applies Theorem \ref{mgDensityCharac}
 to both models $\left(S^{(1)},Q,\mathbb F\right)$ and $\left(S^{(0)}:=I_{\{Z_{-}<1\}}\is {\cal T}_a(S), \mathbb F\right)$,
  and compare the conditions (\ref{Assump1})-(\ref{Assump2})-(\ref{Assump3}) associated to these models.
  To this end, we need to derive the predictable characteristics, $\left(b^{(1,Q)},c^{(1,Q)}, F^{(1,Q)}(dx),A^{(1,Q)}\right)$,
   of $\left(S^{(1)},Q,\mathbb F\right)$. Then, the $\mathbb F$-compensator of $\mu^{(1)}$ under $Q$ is
   $\nu^{(1,Q)}(dt,dx):=\psi(x)\nu(dt,dx)$ that coincides with $\nu^{(0)}$. Then,  using
  $(hI_{\{\widetilde Z=1>Z_{-}\}}\star\mu)^{p,\mathbb F}=h(1-\psi)I_{\{Z_{-}<1\}}\star\nu$ and that the compensator  under $Q$ of
   $H\star\mu$ --for any nonnegative and $\widetilde{\cal P}(\mathbb F)$-measurable $H$-- is $H(I_{\{\psi=0\}}+\psi)\star\nu$,
   we get
 \begin{eqnarray*}
 b^{(1,Q)}:=b-\int h(x)(\psi(x)-1)F(dx),\ \ \ \ c^{(1,Q)}:=c,\\
  F^{(1,Q)}(dx):=\psi(x)F(dx),\ \ \ \ \ A^{(1,Q)}:=I_{\{Z_{-}<1\}}\is A.
 \end{eqnarray*}
By comparing the above quadruplet to the quadruplet given in (\ref{predCharac4S0}), we conclude that
the two models, $\left(S^{(1)},Q,\mathbb F\right)$ and $\left(S^{(0)}:=I_{\{Z_{-}<1\}}\is {\cal T}_a(S), \mathbb F\right)$,
 have the same predictable characteristics.
  Hence, the proof of the proposition follows immediately from applying Theorem  \ref{mgDensityCharac} and
 using the same pair of $\mathbb F$-predictable functionals (i.e. $(\beta^{(0)},f^{(0)})=(\beta^{(1)},f^{(1)})$),
 as the conditions (\ref{Assump1})-(\ref{Assump2})-(\ref{Assump3}) are exactly the same for both models.



\begin{acknowledgements}
The research of Tahir Choulli  and Jun Deng is supported financially by the
Natural Sciences and Engineering Research Council of Canada,
through Grant G121210818. The research of Anna Aksamit and Monique Jeanblanc is supported
by Chaire Markets in transition, French Banking Federation.\\
The second author, TC, is very grateful to Monique Jeanblanc and LaMME (Evry Val d'Essonne University), where this work started and was completed, for their welcome and their hospitality.
\end{acknowledgements}




\end{document}